\documentclass[twocolumn,english,amsthm]{autartArxiv}   

\usepackage{graphicx}          
\usepackage{epstopdf}          

\usepackage{babel}
\usepackage{natbib}
\usepackage{amsmath,amssymb}
\usepackage{booktabs}
\usepackage{array}
\usepackage[shortcuts]{extdash}
\usepackage{float}

\allowdisplaybreaks

\DeclareMathOperator*{\argmin}{arg\;min}

\newtheorem{mydef}{Definition}
\newtheorem{mythm}{Theorem}
\newtheorem{assumption}{Assumption}

\usepackage{xcolor}
\usepackage{ulem}
\newcommand{\Koen}[2][]{{\color{black}#2}}				
\newcommand{\Maarten}[2][]{{\color{black}#2}}		
\newcommand{\KoenSpellCheck}[2][]{{\color{black}#2}}
\newcommand{\M}[2][]{{\color{black}#2}}

\begin{document}

\begin{frontmatter}

\title{Identification of \Koen[Nonlinear Block-Oriented]{Block-Oriented Nonlinear} Systems starting from Linear Approximations: A Survey} 

\thanks[footnoteinfo]{The corresponding author is M.~Schoukens (maarten.schoukens@vub.ac.be).}
\thanks[preprint]{This paper is a postprint of a paper submitted to and accepted for publication in Automatica. This manuscript version is made available under the CC-BY-NC-ND 4.0 license. The published copy of record is available through: https://doi.org/10.1016/j.automatica.2017.06.044}
\author{Maarten Schoukens, Koen Tiels} 
\address{Vrije Universiteit Brussel, Dept. ELEC, Pleinlaan 2, B-1050 Brussels, Belgium}  

\begin{keyword}                           
System Identification, Maximum Likelihood, Nonlinear Systems, Hammerstein, Wiener, Wiener\-/Hammerstein, Hammerstein\-/Wiener, Parallel Cascade, Feedback, Best Linear Approximation, Linear Fractional Representation
\end{keyword}                             

\begin{abstract}                          
	Block-oriented \Maarten[]{nonlinear} models are popular in nonlinear \M{system identification} because of their advantages of being simple to understand and easy to use. Many different identification approaches were developed over the years to estimate the parameters of a wide range of block-oriented \Maarten[]{nonlinear} models. One class of these approaches uses linear approximations to initialize the identification algorithm. The best linear approximation framework and the $\epsilon$-approximation framework, or equivalent frameworks, allow the user to extract important information about the system, guide the user in selecting good candidate model structures and orders, and \M[they]{} prove to be a good starting point for nonlinear system identification algorithms. This paper gives an overview of the different block-oriented \Maarten[]{nonlinear} models that can be \M[modeled]{identified} using linear approximations, and of the identification algorithms that have been developed in the past. A non-exhaustive overview of the most important other block-oriented \Koen{nonlinear} system identification approaches is also provided throughout this paper.
\end{abstract}
\end{frontmatter}

\section{Introduction}
	
	Nonlinear models are \KoenSpellCheck[often used these days]{often used} to obtain a better insight into the behavior of the system under test, to compensate for a potential nonlinear behavior using predistortion techniques, or to improve plant control performance. Popular nonlinear model structures are, amongst others, nonlinear state space models \citep{Suykens1995,Paduart2008,Schon2011}, NARX and NARMAX models \citep{Leontaritis1985,Billings2013}, and block-oriented \Maarten[]{nonlinear} models \citep{Giri2010}. Due to the separation of the nonlinear dynamic behavior into linear time invariant (LTI) dynamics and the static nonlinearities (SNL), block\-/oriented nonlinear models are quite simple to understand and easy to use.
	
	This paper provides a survey of the identification of block-oriented \Maarten[]{nonlinear} models using input-output data only. The identification algorithms that are discussed in this paper are all initialized using linear approximation frameworks for nonlinear systems: the Best Linear Approximation \Maarten[]{(BLA)} \citep{Enqvist2005,Enqvist2005a,Pintelon2012} and the $\epsilon$-approximation \citep{Schoukens2015} framework. An overview on the practical use of the BLA is provided in \citep{Schoukens2016}, however, \Maarten[]{\KoenSpellCheck[this]{that} paper does not address the use of the BLA for block-oriented \Koen{nonlinear} modeling purposes.}  Many other block-oriented identification approaches, different from the linear approximation approaches, are described in the literature (e.g. overparametrization, iterative, or kernel-based methods). A non-exhaustive overview of the most important approaches for each model structure is given at the beginning of each section discussing that particular model structure. The presented identification approaches and the considered block-oriented structures go beyond the methods and structures that \KoenSpellCheck[are presented in previous surveys, such as \citep{Billings1980,Hasiewicz2010} and \citep{Giri2010}]{previous surveys, such as \citep{Billings1980,Hasiewicz2010,Giri2010}, present}.
	
	Block-oriented identification approaches starting from linear approximations have the advantage \KoenSpellCheck[that]{of addressing} part of the model selection problem \KoenSpellCheck[is done]{} in an LTI-framework, which is better understood by many practitioners. These approaches \Maarten[prove to obtain]{obtain} good results in various benchmark problems (such as the Wiener-Hammerstein \citep{Schoukens2009a} and Silverbox benchmarks \citep{Wigren2013}), as illustrated in Section \ref{sec:Benchmark} of this paper. Linear models of nonlinear systems can also be used to detect the underlying block-oriented structure of the system under test. \KoenSpellCheck[This structure detection problem using linear approximations is discussed in \citep{Haber1990,Lauwers2008,Schoukens2015}, however, no identification approaches are presented there.]{\citet{Haber1990,Lauwers2008,Schoukens2015} discuss this structure detection problem using linear approximations, however, they do not present identification approaches there.} 
		
	\begin{figure*}
		\centering
			\includegraphics[width=1.8\columnwidth]{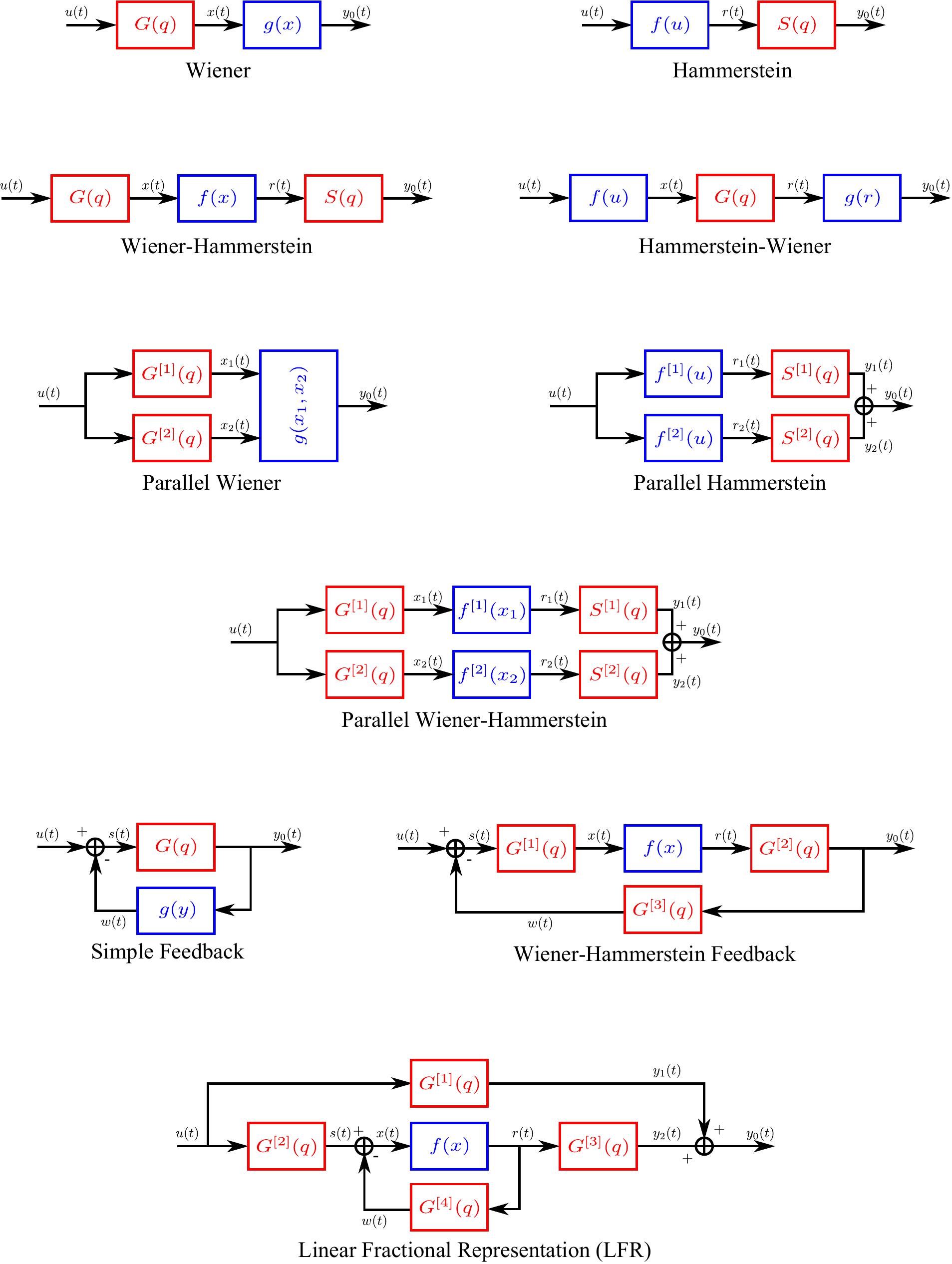}
		\caption{The block-oriented structures that are considered in this paper. The different structures have been obtained by using series, parallel and feedback connections of LTI blocks $G(q)$ and $S(q)$ and static nonlinear blocks $f(\cdot)$ and $g(\cdot)$. There are three types of structure classes: single branch structures (Wiener, Hammerstein, Wiener-Hammerstein and Hammerstein-Wiener), parallel branch structures (parallel Wiener, parallel Hammerstein, parallel Wiener-Hammerstein) and feedback structures (simple feedback structure, Wiener-Hammerstein feedback and LFR).}
		\label{fig:Structures}
	\end{figure*}
	
	\M[This paper introduces first]{} The different block-oriented structures that are \M[considered]{studied} throughout this paper \KoenSpellCheck{are introduced} in Section~\ref{sec:Structures}. Bussgang's theorem is very useful \KoenSpellCheck[to determine]{for determining} the BLA of many \M[of the considered]{} block-oriented structures. This theorem is discussed in Section~\ref{sec:Bussgang}. The considered linear approximation frameworks are discussed next in Section~\ref{sec:BLA}. Section~\ref{sec:ML} \M[discusses the maximum likelihood nature of]{introduces the cost function that is considered throughout this paper.} The identification algorithms for the single branch models (Hammerstein, Wiener, Hammerstein-Wiener and Wiener-Hammerstein) are studied in Section~\ref{sec:Single}. The parallel block-oriented \Koen{nonlinear} modeling approaches are discussed in Section~\ref{sec:Parallel}. Section~\ref{sec:Feedback} considers the identification of different nonlinear feedback models. Section~\ref{sec:Overview} provides an overview of the pros and the cons of the model structures and their identification based on the BLA\M[, as explained in Sections~\ref{sec:Single}--\ref{sec:Feedback}]{}. Some guidelines for the user are presented in Section~\ref{sec:Guidelines}. The good performance of the identification methods \M[that are]{} discussed in this paper is illustrated by some benchmark results in Section~\ref{sec:Benchmark}. Finally, some open research problems are highlighted in Section~\ref{sec:futureWork}.

\section{Block-Oriented \Maarten[]{Nonlinear} Models: Structures, Representation and Noise Framework} \label{sec:Structures}
		
	\subsection*{Model Structures}
		
	Block-oriented \Maarten[]{nonlinear} models are constructed starting from two basic building blocks: a linear time-invariant (LTI) block and a static nonlinear block. They can be combined in many different ways. Series, parallel and feedback connections are considered in this paper, resulting in a wide variety of block-oriented structures as is depicted in Figure~\ref{fig:Structures}. These block-oriented \Maarten[]{nonlinear} models are only a selection of the many different possibilities that one could think of. \KoenSpellCheck[For instance the generalized Hammerstein-Wiener structure that is discussed in \citep{Wills2012} and the Hammerstein/nonlinear feedback model structure that is discussed in \citep{Pelt2001} are not considered in this paper.]{For instance, this paper does not consider the generalized Hammerstein-Wiener structure that \citet{Wills2012} discuss, nor the Hammerstein/nonlinear feedback model structure that \citet{Pelt2001} discuss.} A comparison of the model complexity and the identification procedures of the considered model structures is made in Section~\ref{sec:Overview}.
	
	\subsection*{Model Representation}
	
	The LTI blocks and the static nonlinear blocks can be represented in many different ways. The LTI block could for instance be modeled as a nonparametric \Maarten[FRF]{frequency response function} \citep{Giri2013} or an impulse response \citep{Lacy2003}, or they can be parametrized using a state space \citep{Verhaegen1996}, rational transfer function \citep{Narendra1966}, or basis function expansion \citep{Tiels2014b}. The static nonlinear block can again be represented in a nonparametric way using, for instance, kernel-based methods \citep{Mzyk2014}, or in a parametric way using, for instance, a \KoenSpellCheck[linear-in-the-parameters]{linear in the parameters} basis function expansion (polynomial, piecewise linear, radial basis function network, ...) \citep{SchoukensM2014}, neural networks \citep{SchoukensM2015a}, or other dedicated \KoenSpellCheck[parameterizations]{parametrizations} for static nonlinear functions.
	
	The methods that are presented in this paper will typically use, but are not limited to, a rational transfer function representation for the LTI blocks and a \KoenSpellCheck[linear-in-the-parameters]{linear in the parameters} basis function expansion for the static nonlinearity.

	Another issue of block-oriented \Maarten[]{nonlinear} models is the uniqueness of the model parametrization. Gain exchanges, delay exchanges and equivalence transformations are present in many block-oriented structures \citep{SchoukensM2015a,SchoukensM2015c}. This results in many different models with the same input-output behavior, but with a different parametrization. The Jacobian of the cost function can be rank deficient due to these indistinguishability issues that are common for block-oriented identification problems. The related degenerations in the Jacobian need to be taken into account during the optimization.
	
	The indistinguishability issues do not only impact \M{on} the parametrization of the chosen model structure, \M[it is]{they} also \M[present between]{affect the} model structures \M{themselves}. The simple feedback structure (see Figure~\ref{fig:Structures}) can both be represented using an LTI block in the forward path and a static nonlinearity in the feedback path\M{,} or the other way around, a static nonlinearity in the forward path and an LTI block in the feedback path as is depicted in Figure~\ref{fig:IndistFB}. Both \M[representations]{model structures} lead to the same input-output representation, but are structurally different \citep{Schoukens2008}. Indeed\KoenSpellCheck{,} we have that the outputs of the two possible simple feedback structures are given by:
	\begin{align}\begin{cases}
		y(t) = G(q)\left[ u(t) - g(y(t)) \right], \\
		\bar{y}(t) = \bar{g}\left( u(t) - \bar{G}(q)\bar{y}(t) \right).
	\end{cases}\end{align}
	\Maarten[which can be rewritten as:
	\begin{align}\begin{cases}
		u(t) = G^{-1}(q)y(t) + g(y(t)), \\
		u(t) = \bar{g}^{-1}(\bar{y}(t)) + \bar{G}(q)\bar{y}(t).
	\end{cases}\end{align}]{}
	Both representations should have the same input-output \M{relationship}, which results in the following \Maarten[system of equations]{equation}:
	\Maarten[\begin{align}\begin{cases}
		y(t) = \bar{y}(t), \\
		G^{-1}(q)y(t) + g(y(t)) = \bar{g}^{-1}(y(t)) + \bar{G}(q)y(t).
	\end{cases}\end{align}]{\begin{align} G^{-1}(q)y(t) + g(y(t)) = \bar{g}^{-1}(y(t)) + \bar{G}(q)y(t).
    \end{align}}
	It can be easily observed that this \M{equation} holds for $\bar{G}(q) = G^{-1}(q)$ and $\bar{g}(.) = g^{-1}(.)$. \Maarten{It is assumed here that both $G(q)$ and $g(y)$ are invertible. Note that the equivalence transformation can introduce an unstable (when $G(q)$ is not minimal-phase) and non-causal LTI block (when $G(q)$ contains a pure delay).}
	
	\begin{figure}
		\centering
			\includegraphics[width=0.95\columnwidth]{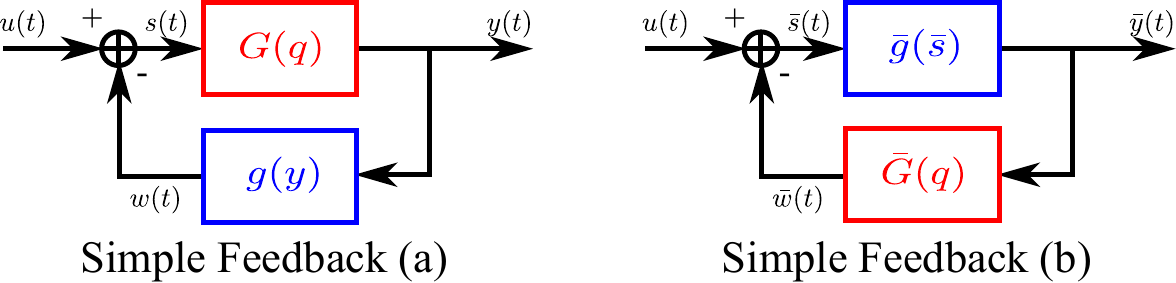}
		\caption{Two simple feedback block-oriented \Koen{nonlinear} model structures: (a) LTI block $G(q)$ in the forward path and a static nonlinearity $g(.)$ in the feedback path, (b) a static nonlinearity $\bar{g}(.)$ in the forward path and an LTI block $\bar{G}(q)$ in the feedback path. Both structures share the same input-output behavior when $\bar{g}(.) = g^{-1}(.)$ and $\bar{G}(q) = G^{-1}(q)$.}
		\label{fig:IndistFB}
	\end{figure}
	
	\subsection*{Noise Framework}
	
	\begin{assumption} \label{ass:Noise}
		\textbf{Noise framework:}
		A Gaussian additive, colored zero-mean noise source $n_y(t)$ with a finite variance $\sigma^2$ is present at the output of the system only:
		\begin{align}
			y(t) = y_0(t) + n_y(t). \label{eq:noise}
		\end{align}
		This noise $n_y(t)$ is assumed to be independent of the known input $u(t)$. The signal $y(t)$ is the actual output signal and a subscript $0$ denotes the exact (unknown) value.
	\end{assumption}
	
	The choice \KoenSpellCheck[for]{of} the noise framework that is described in Assumption~\ref{ass:Noise} is a simplified representation of reality. This simplification can lead to biased estimates when there are other noise sources present, located at other positions inside the system, e.g. process noise passing through a nonlinear subsystem \citep{Hagenblad2008}. A more realistic noise framework can be obtained by introducing multiple noise sources, or by placing the noise source at a different location in the considered system structure. \KoenSpellCheck[Multiple noise sources are, for instance, considered in \citep{Hagenblad2008,Wills2013,Lindsten2013,Wahlberg2014}.]{For instance, \citet{Hagenblad2008,Wills2013,Lindsten2013,Wahlberg2014} consider multiple noise sources.} This more realistic noise framework comes often at the cost of a more complex identification algorithm.

\section{Bussgang's Theorem and Separable Processes}	\label{sec:Bussgang}
	\subsection{Bussgang's Theorem}
		Bussgang's theorem \citep{Bussgang1952} is stated in \citep{Papoulis1991} as follows:
		\begin{mythm}
			If the input to a memoryless system $y=g(u)$ is a zero-mean stationary Gaussian input $u(t)$, the cross-correlation of $u(t)$ with the resulting output $y(t) = g(u(t))$ is proportional to $R_{uu}(\tau)$:
			\begin{align}
				R_{uy}(\tau) = \alpha R_{uu}(\tau) \quad \text{where} \quad \alpha = E\{g'(u(t))\}
			\end{align}
		\end{mythm}
		\begin{proof}
			See \cite{Bussgang1952,Papoulis1991}.
		\end{proof}
				
		In the frequency domain this becomes:
		\begin{align}
			S_{YU}(e^{j\omega T_s}) &= \alpha S_{UU}(e^{j\omega T_s}),
		\end{align} 
		\Maarten[]{where $T_s$ is the sampling period, $S_{YU}$ and $S_{UU}$ denote the crosspower spectrum of $y(t)$ and $u(t)$ and autopower spectrum of $u(t)$\Koen{,} respectively.}
        
		Bussgang's theorem proves to be very valuable for the analysis of block-oriented structures using linear approximations \M{as is shown} in the remainder of this paper.
		
	\subsection{Riemann Equivalence Class of Asymptotically Normally Distributed Excitation Signals}	
		Here we consider the Riemann equivalence class of asymptotically normally distributed excitation signals $\mathbb{S}_{U}$ \citep{Schoukens2009,Pintelon2012}. This signal class \M[mostly]{mainly} contains Gaussian noise sequences, but considering the Riemann class allows one to extend \Maarten[$\mathbb{S}_{U}$]{the signal class} to also contain periodic signal sets, while Bussgang's Theorem still applies \citep{Pintelon2012}. 
		
		\begin{mydef} \label{def:Gaussian}
			\textbf{Riemann equivalence class of asymptotically normally distributed excitation signals $\mathbb{S}_{U}$.}
			Consider a signal $u(t)$ with a piecewise continuous power spectrum $S_{UU}(e^{j\omega T_s})$, with a finite number of discontinuities. A random signal $u(t)$ belongs to the Riemann equivalence class if it obeys any of the following statements:
			\begin{enumerate}
				\item $u(t)$ is a Gaussian noise signal with power spectrum $S_{UU}(e^{j\omega T_s})$.
				\item $u(t)$ is a random multisine or random phase multisine \citep{Pintelon2012} such that:
					\begin{align}
						\frac{1}{N} \sum_{k = k_1}^{k_2} E\left\{ \left| U\left( k \right) \right|^2 \right\} &= \frac{1}{2 \pi} \int_{\omega_{k_1}}^{\omega_{k_2}} S_{UU}\left( e^{j \nu T_s} \right) d\nu  \nonumber \\
																																																	& \; + O\left( N^{-1} \right), \nonumber
					\end{align}
					with $\omega_k = k \frac{2 \pi f_s}{N}$, $k \in \mathbb{N}$, $0<\omega_{k_1}<\omega_{k_2}<\pi f_s$, $f_s$ is the sampling frequency, and $U(k)$ is the discrete Fourier spectrum of $u(t)$. \Maarten[]{$E\left\{ . \right\}$ denotes the expected value operator.}
			\end{enumerate}		
		\end{mydef}
			
		A random phase multisine $u(t)$ is a periodic signal with period length $\frac{N}{f_s}$ defined in \citep{Pintelon2012} as: 
			\Maarten[\begin{align}
				u(t) &= N^{-1/2}\sum_{k=-N/2+1}^{N/2-1}U_{k}e^{j(2\pi k\frac{t}{N} f_s+\varphi_{k})},\quad k\neq0	\nonumber \\
			 			 &= N^{-1/2}\sum_{k=1}^{N/2-1}2U_{k}\cos(2\pi k\frac{t}{N} f_s+\varphi_{k}),
			\end{align}]{
            \begin{align}
				u(t) &= N^{-1/2}\sum_{k=1}^{N/2-1}2U_{k}\cos(2\pi k\frac{t}{N} f_s+\varphi_{k}).
			\end{align}
            }
		\Maarten[where $U_{-k}=U_{k}$ and $\varphi_{-k}=-\varphi_{k}$ to obtain a real-valued signal $u(t)$.]{} The phases $\varphi_{k}$ are random variables that are independent over the frequency and are a realization of a random process on $[0, 2\pi[$, such that $E\{e^{j\varphi_{k}}\}=0$.  For instance, the random phases can be uniformly distributed over the interval $[0, 2\pi[$. The (\Maarten[]{real-valued}) amplitude $U_{k}$ is set in a deterministic way by the user and is uniformly bounded by $M_U$ ($0 \leq U_k \leq M_U < \infty$). Random phase multisines have the advantage of being periodic signals. This avoids the adverse effects of spectral leakage for a proper choice of the period length $\frac{N}{f_s}$. They also offer full control over the applied amplitude spectrum to the user. The Riemann equivalence ensures that a random phase multisine is asymptotically ($N \rightarrow \infty$) Gaussian distributed \citep{Pintelon2012}.
	
	\subsection{Separable Processes}
		Bussgang's theorem has been extended to other classes of signals, besides the Gaussian class, using the concept of separable processes that is introduced in \citep{Nuttall1958}. Gaussian processes, sine wave and phase modulated signals are shown to be separable processes in \citep{Nuttall1958}. Furthermore, \citep{McGraw1968} have shown that the signals belonging to the class of elliptically symmetric distributions are separable. \Maarten[]{Also random phase multisines with flat amplitude spectra are proven to be separable in \citep{Enqvist2011}.} A more \M{in-depth} discussion of separable processes for nonlinear system identification can be found in \Koen{\citep{Billings1978b,Enqvist2005a,Enqvist2010}}.

\section{Linear Approximations of Nonlinear Systems} \label{sec:BLA}
	Most real-world systems do not behave completely linearly. Nevertheless, a linear model often explains a significant part of the behavior of a (weakly) nonlinear system. This approximate linear model also provides the user with a better insight into the behavior of the system under test. It motivates the use of a framework that approximates the behavior of a nonlinear system by a linear time invariant model under well-chosen system-specific boundary conditions. This paper uses the BLA framework \citep{Ljung2001,Enqvist2005,Enqvist2005a,Pintelon2012, Schoukens2016} and the $\epsilon$-approximation \citep{Schoukens2015} to estimate a linear approximation of a nonlinear system\Maarten[given a fixed set of such boundary conditions]. This section \KoenSpellCheck[gives a brief introduction to]{briefly introduces} the theoretical framework of the BLA and the $\epsilon$-approximation. 
	
	\subsection{System Class} \label{sec:SystemClass}
		Both the BLA and the $\epsilon$-approximation consider the PISPO (periodic in, \Maarten{the} same period out) system class. This class of nonlinear systems includes systems whose output can be approximated arbitrarily well in mean square sense by a \Koen{uniformly bounded} Volterra series (see \Koen{\citet{Schoukens1998,Schetzen2006,Pintelon2012}} for more details). \Koen{Systems belonging to this class respond to a periodic input with a periodic steady-state output with the same period \citep{Boyd1984}, hence the \KoenSpellCheck[name]{acronym} PISPO.} This system class is linked with the bounded-input bounded-output (BIBO) stability criterion \citep{SchoukensM2015c}. Nonlinear phenomena generating sub-harmonics such as chaos and bifurcations are excluded. Saturation, clipping, and dead zone nonlinearities are, for example, still included in this system class.
		
	\subsection{Best Linear Approximation}
		The theoretical framework of the BLA results in a model whose structure is shown in Figure~\ref{fig:BestLinearApprox}. It includes four components: one boundary condition \Maarten[]{being an} input excitation class $\mathbb{U}$, and three model constituents, namely a Linear Time Invariant (LTI) model labeled $G_{bla}(q)$, a perturbation noise source labeled $n_{y}(t)$, and a nonlinear distortion source labeled $\tilde{y}_{s}(t)$.
			
		\begin{figure}
			\centering
				\includegraphics[width=0.95\columnwidth]{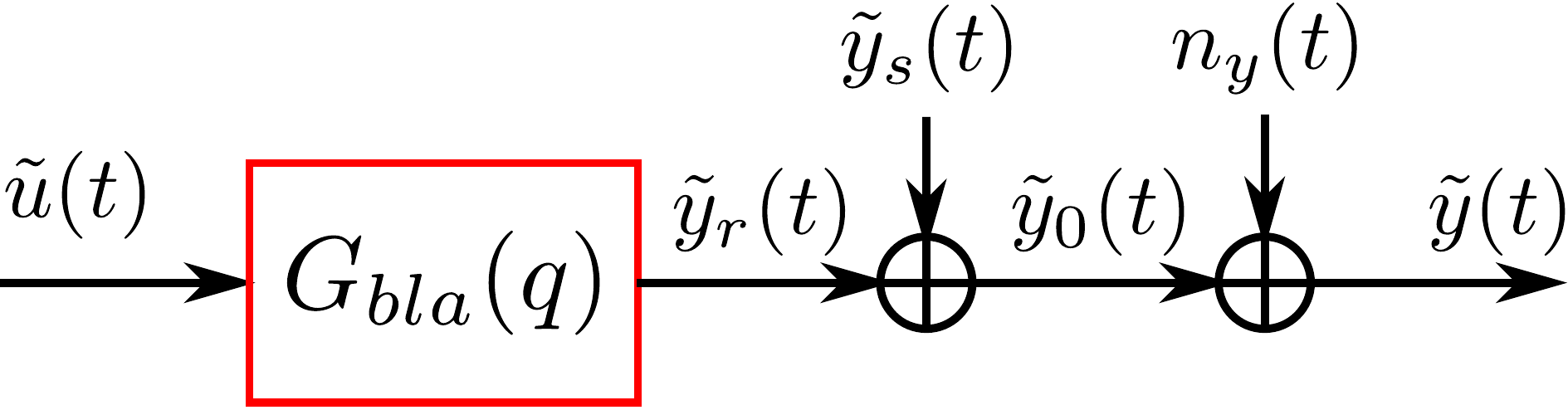}
			\caption{The BLA of a nonlinear system for a given class of excitation signals $\mathbb{U}$ consists of the resulting LTI model $G_{bla}(q)$, the unmodeled nonlinear contributions $\tilde{y}_{s}(t)$, and the additive noise source $n_{y}(t)$ that is assumed to be present at the output of the nonlinear system. The zero-mean input excitation $\tilde{u}(t)$ belongs to the signal class $\mathbb{U}$, \Koen{$\tilde{u}(t)$ and $\tilde{y}(t)$ are defined in eqs.~\eqref{eq:uTilde} and~\eqref{eq:yTilde}}.}
			\label{fig:BestLinearApprox}
		\end{figure}
		
		As mentioned above, the BLA of a system depends on the chosen signal class $\mathbb{U}$. Most identification methods that are discussed in this work consider $\mathbb{U}$ to be the Riemann equivalence class of asymptotically normally distributed excitation signals $\mathbb{S}_{U}$. When this class of signals is considered to provide an excitation signal, the BLA of many block-oriented \Koen{nonlinear} systems becomes, due to Bussgang's Theorem (see Section~\ref{sec:Bussgang}), a simple function of the linear dynamics that are present in that system. All signals are assumed to be stationary for the remainder of this work.
	
		The BLA model of a nonlinear system is an LTI approximation of the behavior of that system. It is best in mean square sense for a fixed class of input signals $\mathbb{U}$ only. The BLA is defined in \citep{Enqvist2005,Enqvist2005a,Pintelon2012} as:
			\begin{align}
		 	  G_{bla}(q) &= \underset{G(q)}{\argmin} \: E\left\{ \left| \tilde{y}(t) - G(q)\tilde{u}(t) \right|^{2} \right\}, \label{eq:BLA}
		 	\end{align}
	 	where $E\left\{.\right\}$ denotes the expected value operator. The expected value $E\left\{.\right\}$ is taken w.r.t. the random input $\tilde{u}(t)$ \Maarten[]{and the noise $n_y(t)$ acting on the output of the system}. The zero-mean signals $\tilde{u}(t)$ and $\tilde{y}(t)$ are defined as:
		 	\begin{align}
		 	  \tilde{u}(t)	&= u(t) - E\left\{ u(t) \right\}, \label{eq:uTilde} \\
		 	  \tilde{y}(t)	&= y(t) - E\left\{ y(t) \right\}. \label{eq:yTilde} 
			\end{align}
		This definition of the BLA is equal to the definition of the linear time invariant second order equivalent model defined in \citep{Ljung2001, Enqvist2005, Enqvist2005a}, when the stability and causality restrictions imposed there are omitted.
		
		It is shown in \citep{Enqvist2005, Enqvist2005a, Pintelon2012} that eq.~\eqref{eq:BLA} is minimized by:
		\begin{align}
			G_{bla}(e^{j\omega T_s}) = \frac{S_{\tilde{Y}\tilde{U}}(e^{j\omega T_s})}{S_{\tilde{U}\tilde{U}}(e^{j\omega T_s})}, \label{eq:BlaPower}
		\end{align}
		where $S_{\tilde{Y}\tilde{U}}$ and $S_{\tilde{U}\tilde{U}}$ denote the crosspower spectrum of $\tilde{y}(t)$ and $\tilde{u}(t)$ and autopower spectrum of $\tilde{u}(t)$ respectively.
		
		The output residuals $n_{t}(t)$ of the model are given by:
			\begin{align}
		 	  n_{t}(t) &= \tilde{y}(t) - G_{bla}(q)\tilde{u}(t)
			\end{align}
		These residuals represent the total distortion that is present at the output of the system. The total distortion can be split in two contributions based on their nature as is depicted in Figure~\ref{fig:BestLinearApprox}. The nonlinear distortion $\tilde{y}_{s}(t)$ represents the unmodeled nonlinear contributions, while the noise distortion $n_{y}(t)$ is the additive noise that is assumed to be present at the system output. The nonlinear distortion and the noise distortion can be calculated separately as:
			\begin{align}
		 	  \tilde{y}_{s}(t) &= \tilde{y}_{0}(t) - G_{bla}(q)\tilde{u}(t), \\
		 	  n_{y}(t) &= \tilde{y}(t) - \tilde{y}_{0}(t),
			\end{align}
		where $\tilde{y}_{0}(t)$ is the unknown noiseless output. The nonlinear distortion $\tilde{y}_{s}(t)$ is uncorrelated with the input $\tilde{u}(t)$, however, $\tilde{y}_{s}(t)$ is not independent of the input $\tilde{u}(t)$ \citep{Pintelon2012}. The noise source $n_{y}(t)$, on the contrary, is \Maarten[typically]{} assumed to be independent of the input $u(t)$.
		
		The obtained model $G_{bla}(e^{j\omega T_s})$ does not only depend on the nonlinear system, but also on the class $\mathbb{U}$ of input signals that is used. The class $\mathbb{U}$ fixes both the probability density function (e.g.\ Gaussian inputs, uniform input distribution, etc.) and the power spectrum (or power spectral density) of the signals that are used to estimate it. Note that the BLA is equal to zero for zero-mean Gaussian signals applied to an even static nonlinear function \citep{Bussgang1952}.
			
		The BLA can be estimated nonparametrically using the so-called robust and fast method \citep{Schoukens2005,Pintelon2012,Schoukens2012}. Afterwards, a parametric model can be estimated \M[on top of]{starting from} the nonparametric one, \M[where]{by using} the nonparametric disturbance model \M[is used]{} as a frequency weighting \citep{Pintelon2012}. One can also \Maarten[choose to estimate]{estimate} a parametric model directly \citep{Ljung1999,Enqvist2005a,Pintelon2012}. This requires the user to make a careful selection of the disturbance model that captures both the nonlinear and noise distortion contributions. \Maarten[]{Note that the estimation of a parametric model also results in a consistent estimate without the use of a disturbance model.}
	
	\subsection{$\epsilon$-Approximation}
	
			The Riemann equivalence class of asymptotically normally distributed excitation signals $\mathbb{S}_{U}$ is very useful to characterize a nonlinear system. \M[There are some drawbacks, however, with the $\mathbb{S}_{U}$ class]{However, the  $\mathbb{S}_{U}$ class has some drawbacks} when it comes to structure detection and system analysis of some \M[classes of systems]{model structures}. It is sometimes more convenient to consider infinitely small signals. For such excitation signals, the linearization of a cascade equals the cascade of the linearization. 
            
            \Maarten[It is clear that such a signal cannot be used in practice]{Measuring infinitely or very small signals in a noisy setting is often not feasible. However, in quite some cases where very high quality measurements can be obtained, e.g. during mechanical vibration tests or while measuring electronic circuits. In these cases very small signals can be applied to the system and can be measured with a sufficiently high signal-to-noise ratio. Deciding whether or not an input signal is small enough to act as an $\epsilon$-excitation is mostly left to the user and the prior knowledge that is available on the system under test. Alternatively, the BLA framework can be used to make this decision (the estimated nonlinear distortion level should be below the estimated noise floor).} Furthermore, the $\epsilon$-excitation class allows formalizing some system properties, especially for the Hammerstein-Wiener and feedback structures considered in this paper. 
			
			The class of $\epsilon$-excitations is defined in \citep{Schoukens2015} as follows:
			\begin{mydef}
				\textbf{Class of $\epsilon$-Excitations $\mathbb{S}_{\epsilon}$}. 
				The signal $u_{\epsilon}(t)$, $t = 1, \dots, N,$ belongs to the class of $\epsilon$-excitations $\mathbb{S}_{\epsilon}$, if it belongs to the class of Riemann equivalent excitation signals $\mathbb{S}_{U}$ and
				\begin{align}
					\sigma^{2}_{u} &= E\{\tilde{u}_\epsilon^2\} = \epsilon^2, \\
					\tilde{u}_\epsilon(t) &= u_\epsilon(t) - E\{u_\epsilon(t)\}.
				\end{align}
			\end{mydef}
			$\epsilon$-excitations are (asymptotically) normally distributed signals with a limited variance. In this paper, results are given for $\epsilon$ converging towards zero.
		
			\begin{mydef}
			 \textbf{$\epsilon$-Approximation:}
			  Consider the best linear approximation $G_{bla}$ obtained for a random excitation $u = u_{DC} + u_\epsilon$, and $u_\epsilon \: \in \: \mathbb{S}_{\epsilon}$, where $u_{DC}$ is a constant offset. Define:
			  \begin{align}
			  	\lim_{\epsilon \rightarrow 0} G_{bla}(e^{j\omega T_s})|_{u_\epsilon \: \in \: \mathbb{S}_{\epsilon}} =  G_{\epsilon}(e^{j\omega T_s}).
			  \end{align}
			\end{mydef}
			
			The definition of $G_{\epsilon}$ can be extended to deterministic signals as discussed and formalized in \citep{Makila2003}. The properties of the $\epsilon$-approximation are discussed in more detail in \citep{Schoukens2015}.
	
\section{\Maarten[Maximum Likelihood]{Cost Function}} \label{sec:ML}
	A least squares cost function is minimized to obtain the model parameters:
	\begin{align} \label{eq:OptimCost}
		\hat{\boldsymbol{\theta}} = \underset{\boldsymbol{\theta}}{\argmin} \: \sum_{t=1}^N \left( y(t)-\hat{y}(t,\boldsymbol{\theta}) \right)^{2},
	\end{align}
	where $\hat{y}(t,\boldsymbol{\theta})$ is the modeled output, depending on the parameters $\boldsymbol{\theta}$. 

	Note that the cost function used in eq.~\eqref{eq:OptimCost} results in a maximum likelihood estimate if white Gaussian additive noise satisfying Assumption~\ref{ass:Noise} is present at the output. A weighted version (using either a nonparametric noise model or a monic parametric noise model) of the cost function needs to be used to obtain a (sample) maximum likelihood estimate in the case of additive colored Gaussian noise. The estimated sample variance of the noise can be used as a weighting function in a frequency domain cost function. A maximum likelihood estimator is asymptotically efficient \citep{Cramer1946,Ljung1999,Pintelon2012}. This means that the maximum likelihood estimator achieves the lowest asymptotic mean squared error possible with a consistent estimator, \Maarten[]{under the assumption that the system belongs to the model class, and using the noise framework specified in Assumption \ref{ass:Noise}}.

	Unfortunately, this cost function is, in most of the cases considered in this paper, non-convex with respect to the parameters $\boldsymbol{\theta}$. A Levenberg-Marquardt algorithm \citep{Levenberg1944,Marquardt1963,Fletcher1991,Pintelon2012} is used to minimize the cost function in a numerically stable way. This algorithm converges to a local minimum of the cost function. Hence, good initial values of the parameters are very important to ensure the convergence of the estimates to the global minimum. Such initial estimates can be obtained using the algorithms that are described in the following sections. Many of these initialization algorithms result in a consistent estimate, but they are not efficient.

\Maarten[]{\section{Initialization Methods: an Overview} \label{sec:Methods}
	Besides the linear approximations-based methods that are described in the following sections, many other initialization methods for block-oriented \Koen{nonlinear} systems are available in the literature. Often two or more estimation methods are combined to identify the system under test. The most important approaches are highlighted here, a more detailed overview can be found in \citep{Giri2010}.
    
    \subsection{Overparametrization}
    	The overparametrization approach is one of the most known methods for the identification of the Hammerstein and Hammerstein-Wiener type of structures  \citep{Chang1971,Hsia1976,Bai1998,SchoukensM2012}.
        
        In a first step\KoenSpellCheck{,} the \KoenSpellCheck[numbers]{number} of parameters is artificially increased to obtain a problem \KoenSpellCheck[which]{that} is \KoenSpellCheck[linear-in-the-parameters]{linear in the parameters}. In the next step, the estimated overparametrized parameter set is reduced, e.g. by using \Koen[an]{a} singular value decomposition to retrieve the underlying \M[low rank]{low-rank} structure of the \KoenSpellCheck[overparameterized]{overparametrized} parameter set.
        
        One of the main downsides of the \KoenSpellCheck[overparameterization]{overparametrization} approach is the large number of parameters that needs to be estimated in the first step, resulting in a higher variance on the estimated parameters. Nuclear-norm regularization and kernel-based methods have been proposed to improve the classical approaches \citep{Falck2010,Risuleo2015}, requiring less user interaction and reducing the variance on the estimated parameters.
           
    \subsection{Alternating Least Squares Methods}
    	Alternating least squares (ALS) is a powerful tool to solve problems \KoenSpellCheck[which]{that} are multilinear in the parameters. ALS was one of the first successful identification methods for the identification of Hammerstein systems \citep{Narendra1966}. Indeed, the Hammerstein identification problem can be written as a problem \KoenSpellCheck[which]{that} is \KoenSpellCheck[bilinear-in-the-parameters]{bilinear in the parameters}. The convergence of this approach is analyzed in detail by \citep{Stoica1981,Bai2004,Bai2010}. Similar techniques have been applied to Hammerstein-Wiener structures \citep{Zhu2002,SchoukensM2012}.
    	
        Alternatively, one can also use separable least squares approaches to estimate the parameters of block-oriented \Koen{nonlinear} models. This technique has been applied by \citep{Westwick2001,Bai2002b,Aljamaan2011} for the identification of Hammerstein and Wiener models.
    
    \subsection{Inverse Estimation}
		Some identification problems that seem to be hard to solve at first glance turn out to be much \KoenSpellCheck[more easy]{easier} to solve when (part of) the system is inverted (if the inverse exists). Such an estimate can serve as a good initialization of the actual system identification problem at hand. \M[The problem of]{} Identifying a Wiener, or a Hammerstein-Wiener system is a well-known example of this approach.
        
        The problem of identifying a Wiener system is nonlinear in the parameters in its most common formulation. However, the problem becomes linear in the parameters when the inverse of the static nonlinearity is identified instead of the static nonlinearity itself \citep{Hagenblad1998,Pajunen1992}. A similar reasoning holds for the Hammerstein-Wiener case \citep{Zhu2002,SchoukensM2012}. Note\KoenSpellCheck[,]{} that this approach results in biased estimates when there is noise present at the output of the system. Nevertheless, the estimates obtained by such an approach can still serve as initial values for a nonlinear optimization.
        
        The inverse estimation approach can also be of interest for nonlinear feedback systems, such as the simple feedback and the Wiener-Hammerstein feedback structure. Indeed, by considering the inverse of the system structure the nonlinear feedback structure can be transformed in a parallel branch feedforward structure \Koen{(see for example Section 4.3.4 in \citet{Markusson2001a}).}

    \subsection{Other Approaches}
    	Besides the approaches discussed above many other block-oriented identification and initialization methods are described in the literature. Some of them are discussed shortly below, while others are included in the following sections when each specific block-oriented nonlinear structure is discussed in more detail.
        
        Subspace methods have been successfully applied for the identification \Koen{of} MIMO \M{(multiple input multiple output)} Hammerstein, Wiener, Wiener-Hammerstein and Hammerstein-Wiener structures \citep{Verhaegen1996,Westwick1996,Goethals2005a,Katayama2016}.
    	
        Frequency-\Koen[D]{d}omain approaches can be used to exploit the nonlinear behavior of the block-oriented nonlinear system\KoenSpellCheck{,} which is not visible in the time-domain, e.g. \citep{Giri2014}.
        
        The structured nature of block-oriented nonlinear systems can be exploited by carefully designing the input signal used during the estimation, e.g. step signals, sine wave signals, sinesweeps or phase-coupled multisines \citep{Giri2014,Tiels2015,Rebillat2016,Castro2016}. Most of the methods described in the following sections make use of Gaussian input signals (Definition~\ref{def:Gaussian}) to exploit Bussgang's Theorem (see Section~\ref{sec:Bussgang}).
        
    	The problem of identifying a block-oriented \Koen{nonlinear} system can sometimes be written as a convex optimization problem, or the problem can be relaxed such that it becomes convex. An example of such an approach is described in \citep{Cai2011,Han2012}.
}
\section{Single Branch Models} \label{sec:Single}

The BLA of a single branch block-oriented \Koen{nonlinear} system (Hammerstein, Wiener, Wiener-Hammerstein) provides valuable information about the system. It is shown in the following sections that the BLA of these systems is equal to a scaled version of the linear dynamics that are present in the system. This allows the user to perform an important part of the model selection problem, i.e. the selection of the number of poles and zeros in the system dynamics, in the LTI framework. Although the case of the Hammerstein-Wiener system structure is a bit more involved, the BLA approach \M[proves to provide]{provides} a good initial estimate for Hammerstein-Wiener identification.

\subsection{Wiener, Hammerstein \& Wiener-Hammerstein \Maarten[]{Structure}} \label{sec:Hammerstein}
	
    \subsubsection*{Hammerstein Structure} 
    The Hammerstein structure is, together with the Wiener structure, one of the most simple block-oriented \Koen{nonlinear} model structures. A Hammerstein structure \citep{Hammerstein1930} consists of a static nonlinear block \Maarten[]{$f(u)$} followed by an LTI block \Maarten[]{$S(q)$} (see Figure~\ref{fig:Structures}). The Hammerstein structure is used to model nonlinear systems for which the static nonlinearity is present at the input of the system, such as a nonlinear actuator followed by a linear process, or some chemical processes and physiological systems \citep{Giri2010}. 
		
	Many different Hammerstein identification algorithms exist in the literature and are listed in \citep{Giri2010}. The first Hammerstein identification algorithm was introduced in \citep{Narendra1966}, and further improved in \citep{Bai2004} and \citep{Bai2010}. Several other Hammerstein identification algorithms were developed in the last decades. A non-exhaustive list is given below where the methods are classified depending on their properties: kernel-based  and mixed parametric-nonparametric identification algorithms \citep{Hasiewicz2010,Mzyk2014,Risuleo2015}, parametric approaches \citep{Chang1971,Crama2004,Schoukens2007}, overparametrization approaches \citep{Bai1998,Falck2010,Risuleo2015}, blind identification algorithms \citep{Bai2002,Vanbeylen2008}, a set-membership approach \citep{Sznaier2009}, and (continuous-time) instrumental variables approach \citep{Laurain2012}. The feedback setting is considered in \citep{Giri2015} and the errors-in-variables setting in \citep{Mu2015}.  The MIMO Hammerstein case is considered in \citep{Verhaegen1996,Goethals2005b}. The works of \citep{Giri2011,Wang2012,Yong2015} extend the Hammerstein structure to contain dynamic nonlinearities (e.g. backlash or hysteresis nonlinearities), and \citep{Tehrani2015,Guarin2015} consider a time-varying Hammerstein structure.
    
    \subsubsection*{Wiener Structure} 
    A Wiener structure \citep{Wiener1958} consists of an LTI block \Maarten[]{$G(q)$} followed by a static nonlinear block \Maarten[]{$g(x)$} (see Figure~\ref{fig:Structures}). A Wiener structure is often used to model systems where the nonlinear behavior is present at the output of the system. Some examples of such structures are systems with sensor nonlinearities, overflow valves and some physiological systems \citep{Giri2010}.
	
	\KoenSpellCheck[Several Wiener system identification algorithms are listed in \citep{Giri2010}.]{\citet{Giri2010} list several Wiener system identification algorithms.} A wide variety of Wiener identification algorithms have been developed in the last decades. Some are grouped here based on their properties. A non-exhaustive list of methods \M[contains]{includes}: nonparametric (or semi-parametric) identification algorithms \citep{Hasiewicz2010,Lindsten2013,Mzyk2014}, parametric approaches \citep{Billings1977,Hunter1986,Wigren1993,Crama2001,Westwick2003}, a minimal Lipschitz approach \citep{Pelckmans2011}, (orthogonal) basis function expansion approaches \citep{Lacy2003,Aljamaan2014}, blind identification algorithms \citep{Vanbeylen2009}, a recursive approach \citep{Greblicki2001}, set-membership approaches \citep{Sznaier2009,Cerone2015}, and \M{approaches} for systems that contain dynamic nonlinearities \citep{Giri2014,Yong2015}. \KoenSpellCheck[The MIMO Wiener system case is considered in \citep{Westwick1996,Janczak2007}.]{\citet{Westwick1996,Janczak2007} consider the MIMO Wiener system case.} Most of the approaches consider that the noise source is present at the output of the system only, however some methods allow for process noise to be present in between the linear dynamics and the \M{static} nonlinearity \citep{Hagenblad2008,Lindsten2013,Wahlberg2014,Wahlberg2015}.
    
    \subsubsection*{Wiener-Hammerstein Structure} 
    A Wiener\-/Hammerstein model has a structure that is a bit more involved than the Wiener or the Hammerstein model structure. The static nonlinearity \Maarten[]{$f(x)$} is now sandwiched \Maarten[in between]{between} two LTI blocks \Maarten[]{$G(q)$ and $S(q)$} (see Figure~\ref{fig:Structures}). The presence of the two LTI blocks results in a problem that is harder to identify: the main challenge in identifying a Wiener-Hammerstein system lies in the separation of the dynamics over the front and the back LTI block. 
		
	The approaches described in \citep{Vandersteen1997,Tiels2015} estimate the two LTI blocks in a nonparametric way using carefully designed input signals. \KoenSpellCheck[Volterra (and \M{tensor} decomposition) based approaches are presented in \citep{Weiss1998,Tan2002,Kibangou2006,Ahmed2013},]{\citet{Weiss1998,Tan2002,Kibangou2006,Ahmed2013} present Volterra (and Tensor decomposition) based approaches.} \KoenSpellCheck[an iterative approach is discussed by \citep{Voros2007},]{\citet{Voros2007} discusses an iterative approach.} \KoenSpellCheck[a recursive EIV method is presented in \citep{Mu2014} and]{\citet{Mu2014} present a recursive EIV method, and} \KoenSpellCheck[evolutionary-based approaches are presented in \citep{Dewhirst2010,Naitali2016}.]{\citet{Dewhirst2010,Naitali2016} present evolutionary-based approaches.} \KoenSpellCheck[The MISO (multiple input single output) Wiener\-/Hammerstein case is discussed in \citep{Boutayeb1995}, and the MIMO case in \citep{Katayama2016}.]{\citet{Boutayeb1995} discuss the MISO (multiple input single output) Wiener\-/Hammerstein case, \citep{Katayama2016} the MIMO case.} Many approaches use the BLA, or a similar correlation analysis, as a starting point for the algorithm e.g.  \citep{Billings1978,Billings1980,Korenberg1986,Westwick2003,Crama2005,Lauwers2011,Tan2012,Sjoberg2012,Westwick2012,SchoukensM2014,Vanbeylen2014,Tiels2014a,Tiels2015,Giordano2015,Ase2015,Katayama2016}.
	
	\subsubsection*{Best Linear Approximation}
        The BLA of a Wiener-Hammerstein system is a scaled version of the LTI dynamics that are present in $G(q)S(q)$ when input signals belonging to \Maarten[the Riemann equivalence class of asymptotically normally distributed excitation signals]{$\mathbb{S}_{U}$} are used. 
		
		\begin{mythm} \label{theo:BlaWienerHammerstein}
			The BLA $G_{bla}(e^{j\omega T_s})$ of a Wiener-Hammerstein system excited by an input signal belonging to \Maarten[the Riemann equivalence class of asymptotically normally distributed excitation signals]{$\mathbb{S}_{U}$} (Definition~\ref{def:Gaussian}) is asymptotically given by \Maarten[]{(for the number of data growing to infinity)}:
			\begin{align}
				G_{bla}(e^{j\omega T_s}) = \alpha G(e^{j\omega T_s})S(e^{j\omega T_s}), \quad \alpha \in \mathbb{R}. \label{eq:BlaWienerHammerstein}
			\end{align}	
		\end{mythm}
		\begin{proof}
			See \citep{Enqvist2010,SchoukensM2015c}.
		\end{proof}
        
\Koen{        For the Hammerstein and the Wiener case, the BLA simplifies to:
        \begin{align}
        	G_{bla,H}(e^{j\omega T_s}) &= \alpha S(e^{j\omega T_s}), \quad \alpha \in \mathbb{R}, \label{eq:BlaHamm} \\
        	G_{bla,W}(e^{j\omega T_s}) &= \alpha G(e^{j\omega T_s}), \quad \alpha \in \mathbb{R}. \label{eq:BlaWiener}
        \end{align}
		}
        
		Note that the gain $\alpha$ \M{not only depends} on the system characteristics, but also on the \Koen{variance}, the coloring, and the offset (DC value) of the input signal. In the case of the Hammerstein structure, the BLA does not depend on the coloring of the input signal since the input signal is applied directly to the nonlinearity of the system without prior filtering.
	
	\subsubsection*{Hammerstein identification using the BLA}
		Many different approaches use the BLA, or similar correlation methods, to identify a Hammerstein system \citep{Billings1978,Hunter1986,Westwick2003,Crama2004}. The identification of a Hammerstein system using the BLA is quite straightforward. The BLA approach decouples the identification of the LTI block and the static nonlinear block. 
		
		Firstly, the LTI dynamics are estimated using the BLA, this results in a consistent estimate of the LTI block $\hat{S}(q)$, up to an unknown gain exchange $\alpha$ (see eq.~\eqref{eq:BlaHamm}). Secondly, the static nonlinearity is estimated using a parametrization \KoenSpellCheck[which]{that} is linear in the parameters. This results in a problem that is linear in the parameters:
		\begin{align}
			y(t) &= S(q)\left[f(u(t))\right] \\
				 &= \sum_{i=1}^{n_f} \gamma_i S(q)\left[f_i(u(t))\right], \label{eq:HammOutLin}
		\end{align}
		where $n_f$ denotes the number of basis functions $f_i(.)$ that is used to describe the static nonlinearity. Note that eq.~\eqref{eq:HammOutLin} is linear in the parameters $\gamma_i$ when $S(q)$ is replaced by its estimate $\hat{S}(q)$. Hence, the parameters $\gamma_i$ can easily be obtained by a linear least squares estimation. This results in a consistent, but not efficient, estimator of the model parameters \citep{SchoukensM2015c}. An efficient estimate can be obtained with the nonlinear optimization described in Section~\ref{sec:ML}.
        
        \subsubsection*{Wiener identification using the BLA}
		The identification of Wiener systems using the BLA is \M[again]{}very similar to the Hammerstein case. \KoenSpellCheck[BLA or correlation based identification methods for Wiener systems are presented in \citep{Billings1978,Hunter1986,Crama2001,Westwick2003}.]{\citet{Billings1978,Hunter1986,Crama2001,Westwick2003} present BLA- or correlation-based identification methods for Wiener systems.}
		
		A two-step identification approach results in a consistent estimate of the model parameters. First, the LTI block is estimated using the BLA. Second, the static nonlinear block is identified starting from the estimated intermediate signal $\hat{x}(t) = \hat{G}(q)\left[u(t)\right]$ and the system output. One can use his/her favorite model structure (polynomial, neural network, ...) to describe the static nonlinearity. This results in a consistent, but not efficient, estimator of the model parameters \citep{SchoukensM2015c}. An efficient estimate can be obtained with the nonlinear optimization described in Section~\ref{sec:ML}.
        
		\subsubsection*{Wiener-Hammerstein identification using the BLA}	
		There are many BLA-based Wiener-Hammerstein identification methods. In this section\KoenSpellCheck{,} we focus on the method that is presented in \citep{Sjoberg2012}. The challenge of separating the front and the back dynamics is reformulated as a pole and zero allocation problem: which poles and zeros belong to the front dynamic block and which poles and zeros belong to the back dynamic block?
		
		The BLA is estimated and parametrized in a first step, and the poles and zeros of the BLA are computed. These poles and zeros are, due to Theorem~\ref{theo:BlaWienerHammerstein} and under the assumption that the product $G(q)S(q)$ does not contain any pole-zero cancellations, the poles and zeros of the LTI blocks of the Wiener-Hammerstein system:
		\begin{align}
				G_{bla} &= \alpha G(q)S(q) = \alpha \frac{\prod_{i=1}^{n_b+n_d}(z_i-q^{-1})}{\prod_{i=1}^{n_a+n_c}(p_i-q^{-1})},
			\end{align}
            \Maarten[]{where $n_a$ and $n_b$ are the number of poles and zeros in $G(q)$ respectively, while $n_c$ and $n_d$ are the number of poles and zeros in $S(q)$.} The poles $p_i$ and the zeros $z_i$ of the parametrized BLA are the combined poles and zeros of the LTI blocks $G(q)$ and $S(q)$ of the Wiener-Hammerstein system.
	
			These poles and zeros of the BLA need to be assigned to either $G(q)$ or $S(q)$: 
			\begin{align}
				G_{bla} &= \alpha \hat{G}(q)\hat{S}(q),	\\	
				\hat{G}(q) &= \frac{\prod_{i=1}^{n_b+n_d}(z_i-q^{-1})^{\beta_{z_i}}}{\prod_{i=1}^{n_a+n_c}(p_i-q^{-1})^{\beta_{p_i}}}, \\
				\hat{S}(q) &= \frac{\prod_{i=1}^{n_b+n_d}(z_i-q^{-1})^{1-\beta_{z_i}}}{\prod_{i=1}^{n_a+n_c}(p_i-q^{-1})^{1-\beta_{p_i}}},
			\end{align}
			 $\beta_{p_i}$ and $\beta_{z_i}$ are binary parameters that assign the pole or zero to either $\hat{G}(q)$ ($\beta_{p_i} = 1$), or to $\hat{S}(q)$ ($\beta_{p_i} = 0$). The parameters $\beta_{p_i}$ and $\beta_{z_i}$ are grouped in the binary parameter vector $\boldsymbol{\beta} \in \{0,1\}^{(n_a + n_b + n_c + n_d) \times 1}$.
			 
		There is a finite number of different realizations of the parameter vectors $\boldsymbol{\beta}$. All possible realizations of the parameter vector are evaluated. To evaluate a realization of $\boldsymbol{\beta}$, a static nonlinearity is estimated for every realization $k$. This problem is linear in the parameters if the static nonlinearity $\hat{f}_k(x(t))$ is described by a linear combination of basis functions. It is solved using linear least squares regression, resulting in the estimate $\hat{f}_k$. Finally, all the estimated models are ranked based on their mean squares error $mse_k$:
		\begin{align}
		 mse_k &= \frac{1}{N} \sum_{t=1}^{N} \left(y(t) - \hat{y}_k(t)\right)^2, \\
		 \hat{y}_k(t) &= \hat{S}_k(q)\left[ \hat{f}_k\left( \hat{G}_k(q)\left[ u(t)\right] \right) \right],
		\end{align}
		where $N$ is the number of data points in the measured input-output record. The model with the lowest error is selected. 
	
		It is important to note that each complex conjugate pole or zero pair is assigned either to $S(q)$ or $G(q)$\M[in the combinations]. The single elements of the pair are never assigned separately. This is required as the LTI blocks consist of transfer functions with real coefficients only. This also introduces a possible disadvantage of the approach: it can happen that two real poles (or zeros) are combined into a complex conjugate pole (or zero) pair during the parametrization of the BLA since the solution handles real poles and zeros as single elements. This pair cannot be correctly assigned if these poles (or zeros) originate from the two different LTI subsystems.
		
		The total number of realizations of $\boldsymbol{\beta}$ depends on the number of poles and the number of zeros, $n_p$ and $n_z$ respectively. The total number is at least $2^{\frac{n_{p}+n_{z}}{2}}$ and at most $2^{n_{p}+n_{z}}$. The minimum number of realizations is obtained when all poles and zeros are part of a complex conjugate pair, the maximum number of realizations is obtained when all poles and zeros are real. This number increases very rapidly with the model order, which makes the brute-force scan approach in \citep{Sjoberg2012} computationally expensive, since it equals the total number of least squares regressions that need to be performed to scan all possible pole-zero allocations.
		
		A discrete optimization over the binary parameter vector $\boldsymbol{\beta}$ can be used to speed up the algorithm \citep{SchoukensM2014b}.  The evaluation of all possible pole and zero allocations can be avoided by using a fractional approach \citep{Vanbeylen2014,Giordano2015}, or by using a \M{higher-order} correlation, e.g. the so-called QBLA, to separate the dynamics over the front and the back LTI block \citep{Billings1978,Schoukens2008,Westwick2012,SchoukensM2014}. An equivalent subspace version of the method presented in \citep{Sjoberg2012} is developed in \citep{Ase2015}, and extended to the MIMO case in \citep{Katayama2016}.
		
		Although \KoenSpellCheck[it is shown in \citep{Sjoberg2012}]{\citet{Sjoberg2012} show} that this approach results in a consistent estimate of the system parameters, a nonlinear optimization of all the parameters simultaneously as described in Section~\ref{sec:ML}, following the linear initialization step, can be used to reduce the variance of the estimates.

\subsection{Hammerstein-Wiener \Maarten[]{Structure}} \label{sec:HammersteinWiener}
	\M[Hammerstein and Wiener models were discussed in the previous sections. A more general]{Another single branch model} structure is the Hammerstein\-/Wiener model (static nonlinear block \Maarten[]{$f(u)$} - LTI block \Maarten[]{$G(q)$} - static nonlinear block \Maarten[]{$g(r)$}, see Figure~\ref{fig:Structures}). These models are used in a wide range of applications such as chemical processes \citep{Giri2010}, ionospheric dynamics \citep{Palanthandalam2005}, submarine detection \citep{Abrahamson2007}, and RF (radio frequency) power amplifier modeling \citep{Taringou2010}. 
	
	Different identification approaches are proposed to identify a Hammerstein\-/Wiener model or models that are very similar to Hammerstein\-/Wiener block structures. Amongst those are: iterative approaches \citep{Zhu2002, Voros2004}, overparametrization methods \citep{Bai1998,SchoukensM2012,SchoukensM2015c}, frequency domain methods \citep{Crama2004b,Brouri2014a}, subspace methods \citep{Goethals2005a}, stochastic algorithms \citep{Wang2008}, blind approaches \citep{Bai2002}, continuous-time instrumental variable approaches \citep{Ni2013}, and maximum likelihood\KoenSpellCheck{-}based methods \citep{Wills2013}. The identification of Hammerstein-Wiener systems with a backlash input nonlinearity is discussed in \citep{Brouri2014b}.
	
	\subsubsection*{Best Linear Approximation}
		Contrary to Hammerstein and Wiener systems, the BLA of a Hammerstein-Wiener system cannot be simplified to a simple function of the system dynamics, even if input signals belonging to \Maarten[the Riemann equivalence class of asymptotically normally distributed excitation signals]{$\mathbb{S}_{U}$} are used. The intermediate signal $r(t)$ (see Figure~\ref{fig:Structures}) is not Gaussian due to the presence of the first static nonlinearity. Hence, Bussgang's theorem cannot be applied. However, \citep{Wong2012} shows that the BLA will not differ too much from a scaled version of the LTI block $G(q)$ in many practical cases (i.e. if the memory length of $G(q)$ is large w.r.t. the correlation length of the intermediate signal $x(t)$).
		
		The theoretical framework of the $\epsilon$-approximation \M[(a small signal analysis of the system, see \citep{Schoukens2015})]{}can also be used to replace the BLA framework. The $\epsilon$-approximation of a Hammerstein-Wiener system is a scaled version of the LTI dynamics $G(q)$ present in the system \citep{Schoukens2015}.
	
	\subsubsection*{Identification using the BLA}
		 The BLA (or $\epsilon$-approximation) is used as a starting point for the identification of a Hammerstein-Wiener system in \citep{Crama2004b}. First, the BLA is estimated resulting in an estimate $\hat{G}(q)$ of the LTI block. Next, the front static nonlinearity and the inverse of the back static nonlinearity are identified by minimizing:
		\begin{align}
			V_N = \sum_{t=1}^{N} \left( \sum_{i=1}^{n_f} \gamma_i \hat{G}(q)\left[ f_i(u(t)) \right]  - \sum_{j=1}^{n_g} \delta_j g_j(y(t)) \right)^2
		\end{align}
		where it is assumed that the static nonlinearity $f(u)$ and the inverse of the static nonlinearity $g(r)$ can be represented by a linear combination of $n_f$ and $n_g$ nonlinear basis functions respectively. This problem is linear in the parameters, it can be solved using a total least squares approach \citep{VanHuffel2002}. In a next step the estimate of the inverse of the back static nonlinear block can be replaced by an estimate of the forward static nonlinearity $g(r)$. \Maarten[]{Note that the estimates obtained using the total least squares approach are not consistent due to the noise framework specified in Assumption~\ref{ass:Noise}}. These estimates serve as an initialization for the nonlinear optimization described in Section~\ref{sec:ML}\Maarten[]{, resulting in a consistent estimate of the model parameters (if all assumptions are satisfied)}.
        
\section{Parallel Branch Model Structures} \label{sec:Parallel}
\Maarten[]{Parallel block-oriented structures (parallel Hammerstein, parallel Wiener and parallel Wiener-Hammerstein, see Figure~\ref{fig:Structures}) are a logical extension of the single branch block-oriented structures. The use of a parallel branch structure can be motivated by the presence of multiple signal paths inside the system under test. Another possible reason for the use of parallel branch models is their increased modeling capacity. It is well-known that the parallel Wiener and the parallel Wiener-Hammerstein structure are universal approximators of fading memory systems (see Section~\ref{sec:Overview} for more information)}. 

In the case of parallel block-oriented structures the BLA will \M{not only} provide \M[not only]{}information about the LTI blocks that are present in the system, but also about the number of parallel branches in the system. By combining the BLAs obtained at different setpoints of the systems, the selection of the number of parallel branches and the selection of the number of poles and zeros that are present in the system can be performed \citep{SchoukensM2011,SchoukensM2012b,SchoukensM2015a}. This allows the user to perform the selection of the number of poles and zeros in the system dynamics in the LTI framework, and to select the number of parallel branches in the system at an early stage of the identification procedure.

The BLA can also be used to construct a set of orthonormal basis functions to represent the system dynamics \citep{Tiels2014b}. These basis functions are based on a number of user-selected pole locations. \M[The poles of the BLA are a natural choice as a starting point to construct the orthonormal basis functions.]{A natural choice is to use the poles of the BLA as a starting point to construct the orthonormal basis functions.}

\subsection{Parallel Hammerstein and Parallel Wiener \Maarten[]{Structures}} \label{sec:ParallelHammerstein}
	Parallel Hammerstein and parallel Wiener models are generalized forms of the Hammerstein and Wiener structure. They consist of a parallel connection of more than one Hammerstein (or Wiener) system (see Figure~\ref{fig:Structures}). Parallel Hammerstein models are sometimes called Uryson models, memory polynomial, or generalized Hammerstein models in the literature, see for example \citep{Gallman1975,Billings1980,Doyle2002,Ghannouchi2009}. 

	\textit{Parallel Hammerstein:} A popular approach that is used to identify a parallel Hammerstein system is to feed the input signal to a parallel connection of some nonlinear basis functions that are each followed by a finite impulse response filter (FIR) to realize a branch. The corresponding LTI multiple-input-single-output identification problem is then solved using one's favorite FIR identification method in \citep{Gadringer2007}. In earlier work \citep{Gallman1975}, Hermite polynomials were used as orthogonal basis functions to model the static nonlinearity for Gaussian input signals. These approaches result in a linear least squares identification problem, which is easy to solve. However, they suffer from two drawbacks: i) the user gets no physical insight into the number of parallel paths in the device under test (DUT), as the number of parallel branches in the model is set by the degree of the nonlinearity, ii) for systems with long memory effects, a large number of parameters is needed due to the FIR-nature of the dynamic model. The model is therefore not parsimonious. \Maarten{\KoenSpellCheck[A fully nonparametric approach, requiring almost no user interaction, is presented in \citep{Rebillat2016}, based on an exponential sinesweep excitation.]{\citet{Rebillat2016} present a fully nonparametric approach, which requires almost no user interaction, based on an exponential sinesweep excitation.}}
	
	\textit{Parallel Wiener:} Previously proposed parallel Wiener estimation methods suffer from some disadvantages. Some methods rely on an estimate of the Volterra kernel of the system under test \citep{Kibangou2009}. This requires a very large amount of data for the identification. Other methods are limited to the use of finite impulse response models for the linear subsystems \citep{Lyzell2012c,Westwick1997,Korenberg1991}. These approaches typically result in an unwanted high number of parallel branches.
	
	\subsubsection*{Best Linear Approximation}
		Similar to the single branch case, the BLA of a parallel Hammerstein and parallel Wiener system is a simple function of the LTI dynamics that are present in the system.
		
		\begin{mythm} \label{theo:BlaParallelHammerstein}
			The BLA $G_{bla}(e^{j\omega T_s})$ of a parallel Hammerstein system excited by an input signal belonging to \Maarten[the Riemann equivalence class of asymptotically normally distributed excitation signals]{$\mathbb{S}_{U}$} (Definition~\ref{def:Gaussian}) is asymptotically given by \Maarten[]{(for the number of data growing to infinity)}:
			\begin{align}
				G_{bla}(e^{j\omega T_s}) = \sum_{k=1}^{n_{br}} \alpha_k S^{[k]}(e^{j\omega T_s}), \quad \alpha_k \in \mathbb{R}, \label{eq:BlaParallelHammerstein}
			\end{align}	
            \Maarten[]{where $n_{br}$ is the number of parallel branches in the system.}
		\end{mythm}
		\begin{proof}
			The BLA of a Hammerstein system is given by eq.~(\ref{eq:BlaHamm}). This holds for each separate branch in the parallel Hammerstein system. The output $y$ is given by a summation of the outputs $y_k$ of each parallel Hammerstein branch. Hence, the BLA of a parallel Hammerstein system is given by the sum of the BLAs of each parallel branch.
		\end{proof}
		
		\begin{mythm} \label{theo:BlaParallelWiener}
			\Maarten[]{The BLA $G_{bla}(e^{j\omega T_s})$ of a parallel Wiener system excited by an input signal belonging to \Maarten[the Riemann equivalence class of asymptotically normally distributed excitation signals]{$\mathbb{S}_{U}$} is asymptotically given by (for the number of data growing to infinity):
			\begin{align}
				G_{bla}(e^{j\omega T_s}) = \sum_{k=1}^{n_{br}} \alpha_k G^{[k]}(e^{j\omega T_s}), \quad \alpha_k \in \mathbb{R}, \label{eq:BlaParallelWiener}   
			\end{align}	
            under the assumption that the MISO static nonlinearity $g(.)$ of the parallel Wiener system can be approximated arbitrarily well in least squares sense by a MISO polynomial.}
		\end{mythm}
		\begin{proof}
			See \citep{SchoukensM2015c}.
		\end{proof}
		
		An important observation with respect to eq.~(\ref{eq:BlaParallelHammerstein}) and eq.~(\ref{eq:BlaParallelWiener}) is that the \M{input-dependent} gain $\alpha_k$ only appears in the numerator. For example, the BLA in eq.~\eqref{eq:BlaParallelHammerstein} is given by:
		\begin{align}
			G_{bla}(q) = \sum_{k=1}^{n_{br}} \alpha_k S^{[k]}(q) = \sum_{k=1}^{n_{br}} \alpha_k \frac{B_{s}^{[k]}(q)}{A_{s}^{[k]}(q)}.
		\end{align}
		Introducing a common denominator results in the following BLA expression:
		\begin{align}
			G_{bla}(q) &= \frac{\sum_{k=1}^{n_{br}}\alpha_k B_{s}^{[k]}(q) \prod_{j=1,j\neq k}^{n_{br}}A_{s}^{[j]}(q)} {\prod_{k=1}^{n_{br}}A_{s}^{[k]}(q)}.  \label{eq:BlaCommonDen}
		\end{align}
		A similar result can be obtained for the parallel Wiener case.
		
		This means that the poles of the identified BLA are also the poles of the LTI blocks that are present in the system. The zeros of the BLA of a parallel Hammerstein (or parallel Wiener) system may change when the variance\Maarten[, power spectrum,]{} or the offset (DC value) of the input signal changes. \Maarten[]{The gains $\alpha_k$ are also dependent on the coloring of the input signal in the case of a parallel Wiener system.} Indeed, the denominator stays the same when the gains $\alpha_k$ change, but the numerator coefficients depend on the gains $\alpha_k$. \KoenSpellCheck[This is discussed in further detail in \citep{Schoukens2015,SchoukensM2015c,SchoukensM2015a}.]{\citet{Schoukens2015,SchoukensM2015c,SchoukensM2015a} discuss this in further detail.}

	\subsubsection*{Parallel Wiener and parallel Hammerstein identification using the BLA - SVD Approach}
		The main difficulty in identifying a parallel Wiener or parallel Hammerstein system is separating the dynamics over the different parallel branches. This problem can be made easier by allowing more parallel branches to be present in the estimated model (see e.g. \citep{Gallman1975,Billings1980,Doyle2002,Lyzell2012c,Westwick1997,Korenberg1991}). A singular value decomposition (SVD) based approach to identify a parallel Wiener or parallel Hammerstein system with a low number of parallel branches is presented in \citep{SchoukensM2011,SchoukensM2013a,SchoukensM2015a}. This section explains the driving idea of this approach.
				
		The approach proposed in \citep{SchoukensM2013a,SchoukensM2015a} starts with an estimation of the BLA of the considered system for different operating conditions. The different operating conditions have been obtained using input signals with different power spectra. This includes the use of different magnitudes, different offsets, or different coloring of the power spectra.
	
		Next, the measured BLAs are parametrized using a different LTI model for each operating condition. A common denominator model is used for all operating conditions simultaneously:
		\begin{align}
				\hat{G}^{[i_r]}_{bla}\left(q,\hat{\boldsymbol{\theta}}_{bla}\right) &= \frac{\hat{d}_{0}^{[i_r]} + \hat{d}_{1}^{[i_r]}q^{-1} + \ldots + \hat{d}_{n_{d}}^{[i_r]}q^{-n_{d}}}{\hat{c}_{0} + \hat{c}_{1}q^{-1} + \ldots + \hat{c}_{n_{c}}q^{-n_{c}}}, \label{eq:ParBlaCommon}
		\end{align}
		where superscript $[i_r]$ denotes the setpoint of the system. This is indeed possible, as eq. \eqref{eq:BlaCommonDen} assures that the poles of the different measured BLAs are the same. A consistent estimate of the overall dynamics that are present in the nonlinear parallel Hammerstein or parallel Wiener system results. A matrix $\hat{\boldsymbol{D}}$ is constructed \M[containing the stacked]{by stacking the} estimated numerator coefficients of the BLAs at the \Maarten[]{$R$} different operating conditions:
        \begin{align}
        	\hat{\boldsymbol{D}} &= \left[ \begin{array}{cccc}
        		\hat{d}_{0}^{[1]} & \hat{d}_{1}^{[1]} & \ldots & \hat{d}_{n_{d}}^{[1]} \\
        		\hat{d}_{0}^{[2]} & \hat{d}_{1}^{[2]} & \ldots & \hat{d}_{n_{d}}^{[2]} \\
        		\vdots						& \vdots						& \ddots & \vdots								 \\
        		\hat{d}_{0}^{[R]} & \hat{d}_{1}^{[R]} & \ldots & \hat{d}_{n_{d}}^{[R]} \\
        		\end{array} \right]. \label{eq:D}
        \end{align}	
		Starting from the parametrized BLAs, a decomposition of the overall dynamics at the different operating conditions is calculated. It uses the singular value decomposition (SVD) of the $\hat{\boldsymbol{D}}$ matrix: 
		\begin{align}
				\hat{\boldsymbol{D}} &= \boldsymbol{U}_{bla} \boldsymbol{\Sigma}_{bla} \boldsymbol{V}_{bla}^T,
		\end{align}
		where the columns of $\boldsymbol{V}_{bla}$ are the right singular vectors which act as an orthonormal basis for the row space of the $\hat{\boldsymbol{D}}$-matrix, $\boldsymbol{\Sigma}_{bla}$ is a diagonal matrix containing the singular values, and the columns of $\boldsymbol{U}_{bla}$ constitute an orthonormal basis for the column space.
			
		The column vectors in $\boldsymbol{V}_{bla}$ provide an estimate of the numerator coefficients for each branch $k$:
		\begin{align}
			\hat{S}^{[k]}(q) &= \frac{\hat{\delta}_{0}^{[k]} + \hat{\delta}_{1}^{[k]}q^{-1} + \ldots + \hat{\delta}_{n_{d}}^{[k]}q^{-n_{d}}}{\hat{c}_{0} + \hat{c}_{1}q^{-1} + \ldots + \hat{c}_{n_{c}}q^{-n_{c}}},
		\end{align}
		where $\hat{\delta}_{j}^{[k]}$ is the element of the $j$-th row and $k$-th column of the matrix $\boldsymbol{V}_{bla}$. 
		
		The number of parallel branches in the system is obtained based on the estimated rank of the decomposed matrix (based on the singular values obtained in $\boldsymbol{\Sigma}_{bla}$). \Maarten[]{Note that the matrix $\hat{\boldsymbol{D}}$ needs to be whitened before the rank estimation can be performed (see \citep{Rolain1997,SchoukensM2015a} for more information).}
		
		Once an estimate of the LTI blocks of the parallel Hammerstein or parallel Wiener system have been obtained, the estimation of the static nonlinearities in the system boils down to a simple linear least squares problem as in the single branch Hammerstein or Wiener case.
		
		Although \KoenSpellCheck[it is shown in \citep{SchoukensM2015a,SchoukensM2015c}]{\citet{SchoukensM2015a,SchoukensM2015c} show} that this approach results in a consistent estimate of the system parameters, a nonlinear optimization of all the parameters simultaneously as described in Section~\ref{sec:ML} can be used to reduce the variance of the estimates.
		
\subsection{Wiener-Schetzen \Maarten[]{Structure}} \label{sec:WienerSchetzen}
	
	A Wiener-Schetzen model (see Figure~\ref{fig:Wiener-Schetzen}) is a subset of parallel Wiener models where the dynamics are represented by orthonormal basis functions (OBFs) \citep{Heuberger2005}, and the nonlinearity is represented by a multivariate polynomial with the outputs of the OBFs as inputs. The parameters of this model are the polynomial coefficients. Nevertheless, the OBFs should be properly selected.

	Wiener-Schetzen models appear in the literature under various names.	\KoenSpellCheck[In the original ideas of Wiener \citep{Wiener1958}, (continuous-time) Laguerre OBFs are used]{In his original ideas, Wiener uses (continuous-time) Laguerre OBFs \citep{Wiener1958}} and the model is called a Laguerre system (see also \citet{Boyd1985}). \citet{Kibangou2005b} \KoenSpellCheck[uses]{use} discrete-time Laguerre OBFs and \KoenSpellCheck[uses]{use} the name Laguerre-Volterra filters. When using generalized orthonormal basis functions (GOBFs), the names GOB-Volterra filters \citep{Kibangou2005a} and Wiener/Volterra models \citep{daRosa2007} are used, amongst others.

	\begin{figure}
		\centering
		\includegraphics[width=0.45\textwidth]{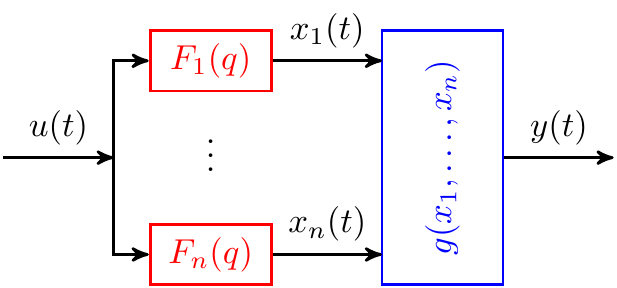}
		\caption{A Wiener-Schetzen model ($F_1(q)$, \ldots, $F_n(q)$ are OBFs, $g(x_1, \ldots, x_n)$ is a multivariate polynomial). \label{fig:Wiener-Schetzen}}
	\end{figure}

	Several methods have been proposed to identify Wiener-Schetzen models. \citet{Gomez2004} \Maarten[proposes]{propose} a non-iterative method, where it is assumed that the nonlinearity is invertible and that this inverse can be described in terms of known basis functions. \KoenSpellCheck[An optimal selection of one of the parameters in a truncated two-parameter Kautz OBF expansion is proposed in \citet{daRosa2007}.]{\citet{daRosa2007} propose an optimal selection of one of the parameters in a truncated two-parameter Kautz OBF expansion.} This optimal selection is done by minimizing an upper bound of the error made by approximating the system's Volterra kernels by the truncated OBF expansion. The Volterra kernels of the system are assumed to be known in this approach. A model structure very similar to a Wiener-Schetzen model is identified in \citet{Totterman2009} and \citet{Gomez2012}, where the nonlinearity is expanded in terms of support vector machines (SVMs).

	\subsubsection*{Best Linear Approximation}
		Since a Wiener-Schetzen model is a parallel Wiener model with a MISO polynomial nonlinearity, its BLA is given by Theorem~\ref{theo:BlaParallelWiener}. In particular, the BLA of a Wiener-Schetzen model is a linear combination of the dynamics in the different branches. The poles of the BLA are thus also the poles of the LTI dynamics that are present in the system.

	\subsubsection*{Identification using the BLA}
	
	OBFs are stable, discrete-time, proper, rational transfer functions $F_i(e^{j \omega})$ that are orthogonal with respect to the inner product
	\begin{equation}
		\langle F_1 , F_2 \rangle = \frac{1}{2 \pi} \int_{- \pi}^{\pi} F_1(e^{j \omega}) {F_2}^*(e^{j \omega}) \mathrm{d}\omega.
	\end{equation}
	The outputs $x_1(t)$ and $x_2(t)$ of two OBFs $F_1$ and $F_2$ that share the same input are orthogonal if the input is white. This follows from the generalized expected value \mbox{$\bar{E}\{x_1(t) x_2(t)\}$}, which can be written as
	\begin{equation}
		\bar{E}\{x_1(t) x_2(t)\} = \frac{1}{2 \pi} \int_{- \pi}^{\pi} F_1(e^{j \omega}) {F_2}^*(e^{j \omega}) \lvert U(e^{j \omega}) \rvert^2 \mathrm{d}\omega
	\end{equation}
	using Parseval's theorem \citep{Heuberger2005}, and which reduces to the inner product of two orthonormal transfer functions if the input is white ($\lvert U(e^{j \omega}) \rvert^2$ is constant). Although in general no claims can be made about the orthogonality for non-white inputs, the considered OBFs are quite robust with respect to the coloring of the OBFs \citep[Chapter~6]{Heuberger2005}.

	Popular choices for the OBFs are Laguerre, Kautz, and Takenaka-Malmquist basis functions \citep{Heuberger2005}. The considered OBFs can be constructed from a set of poles $\{\xi_i\}$, and in its most general form \M{they} are the Takenaka-Malmquist basis functions:
	\begin{equation}
		F_i(q) = \frac{\sqrt{1 - \lvert \xi_i \rvert^2}}{q - \xi_i} \prod_{k = 1}^{i - 1} \frac{1 - {\xi_k}^* q}{q - \xi_k},
	\end{equation}
	for $i = 1, 2, \ldots$; $\xi_i \in \mathbb{C}$; and $\lvert \xi_i \rvert < 1$.
	\M[Note that choosing]{Choosing} all poles $\xi_i$ in the origin results in an FIR model. Choosing all poles identical and real results in a Laguerre basis. A two-parameter Kautz basis is obtained when choosing a pole structure \mbox{$\{ \xi_1, \xi_2, \xi_1, \xi_2, \ldots \}$}. Note that this pole structure can include multiple repetitions of $\xi_1,\xi_2$, where $\xi_1$ and $\xi_2$ are either real or complex conjugated. \M{GOBFs} are obtained when a finite set of $n_{\xi}$ poles is periodically repeated, i.e.
	\begin{equation}
		\xi_{i + (k-1) n_{\xi}} = \xi_i
		, \quad
		i = 1, 2, \ldots, n_\xi
		, \quad
		k = 1, 2, \ldots
	\end{equation}
	
	A proper \Maarten[]{user} choice of the pole locations is important in order to accurately represent the linear dynamics with a limited number of OBFs. Choosing the poles equal to the true poles of the underlying linear dynamics of the (parallel) Wiener system is optimal in this sense. The poles of the BLA are thus excellent candidates to be used in constructing the OBFs. Possible mismatches between the true poles and the estimates from the BLA can be  \Maarten[compensated for]{compensated} by periodically repeating the set of pole estimates. Fast convergence rates \Koen{(convergence in probability of the modeled output to the system output)} can be obtained \citep{Tiels2014b}.
	
	Once the OBFs are constructed, the intermediate signals $x_1(t), \ldots, x_n(t)$ can be computed. Retrieving the polynomial coefficients then boils down to a linear regression \M{problem}. \M[Note that by]{By} choosing OBFs and Hermite polynomials \citep{Schetzen2006}, which are orthogonal for Gaussian inputs, this regression problem is optimally conditioned for white Gaussian inputs.
	Note that this is only true for infinitely long data records; for a finite data record, approximate orthogonality is obtained.
	
	Compared to the BLA/SVD approach in the previous section, only one experiment suffices in this case. The model is also linear in the parameters, but this comes at the cost of a large number of parameters, since the number of polynomial coefficients increases combinatorially with the number of OBFs and the degree of \M{the} nonlinearity. Dimension reduction techniques, such as presented in \citep{SchoukensM2013a,Lyzell2012a,Lyzell2012b,Lyzell2012c} can be used to reduce the number of basis functions, and hence, the number of parameters.
	
\subsection{Parallel Wiener-Hammerstein \Maarten[]{Structure}} \label{sec:ParallelWienerHammerstein}
	
	Parallel Hammerstein and parallel Wiener model structures are restricted to systems with dominant output dynamics (Hammerstein) or input dynamics (Wiener). A further increase in flexibility of the model structure is obtained here by considering the parallel Wiener\-/Hammerstein structure (see Figure~\ref{fig:Structures}). Previously published methods \citep{Baumgartner1975,Wysocki1976,Billings1979,Billings1982} studied a subclass of the parallel Wiener\-/Hammerstein structure. Identification methods based on repeated sine measurements \citep{Baumgartner1975,Wysocki1976}, or white Gaussian inputs \citep{Billings1979} are available for this model structure. \KoenSpellCheck[It is shown in \citep{Palm1978,Palm1979}]{\citet{Palm1978,Palm1979} shows} that a wide class of Volterra systems can be approximated \M{arbitrarily} well using a parallel Wiener\-/Hammerstein model structure. \Maarten[]{\KoenSpellCheck[An identification method based on the decomposition of Volterra kernels is presented in \citep{Dreesen2017}.]{\citet{Dreesen2017} present an identification method based on the decomposition of Volterra kernels.}}
	
	The parallel Wiener\-/Hammerstein system class that is used here is a more general system class than the $S_M$ system class that is used in \citep{Baumgartner1975,Wysocki1976,Billings1979} in the following sense. The $S_M$ model has $M$ parallel branches, and the $m$-th branch contains a monomial nonlinearity that is fixed and equal to $(.)^m$. This restricts the model to have a polynomial nonlinearity only, and to contain only one branch for each degree of this polynomial nonlinearity. Thus a parallel Wiener\-/Hammerstein system containing two parallel branches, each with different LTI subsystems, and with different polynomial nonlinearities can, in general, not be modeled by a $S_M$ model. The methods that \KoenSpellCheck[are presented in \citep{SchoukensM2015a,SchoukensM2015c}]{\citet{SchoukensM2015a,SchoukensM2015c} present} also make some extra assumptions on the parallel Wiener\-/Hammerstein system. However, even when these assumptions are met, the considered system class still allows for a much more complex system behavior than the $S_M$ class does.
	
	\subsubsection*{Best Linear Approximation}
		Similar to the parallel Hammerstein and parallel Wiener case, the BLA of a parallel Wiener-Hammerstein system is a simple function of the LTI dynamics that are present in the parallel branches of the system.		
		\begin{mythm} \label{theo:BlaParallelWienerHammerstein}
			The BLA $G_{bla}(e^{j\omega T_s})$ of a parallel Wiener-Hammerstein system excited by an input signal belonging to \Maarten[the Riemann equivalence class of asymptotically normally distributed excitation signals]{$\mathbb{S}_{U}$} (Definition~\ref{def:Gaussian}) is asymptotically given by \Maarten[]{(for the number of data growing to infinity)}:
			\begin{align}
				G_{bla}(e^{j\omega T_s}) = \sum_{k=1}^{n_{br}} \alpha_k G^{[k]}(e^{j\omega T_s})S^{[k]}(e^{j\omega T_s}), \quad \alpha_k \in \mathbb{R}. \label{eq:BlaParallelWienerHammerstein}
			\end{align}	
		\end{mythm}
		\begin{proof}
			The BLA of a Wiener-Hammerstein system is given by eq.~(\ref{eq:BlaWienerHammerstein}). This holds for each separate branch in the parallel Wiener-Hammerstein system. The output $y$ is given by a summation of the outputs $y_k$ of each parallel Wiener-Hammerstein branch. Hence, the BLA of a parallel Wiener-Hammerstein system is given by the sum of the BLA of each parallel branch.
		\end{proof}
	
	\subsubsection*{Identification using the BLA - SVD Approach}
		The identification algorithm for the parallel Wiener-Hammerstein case is a combination of the identification algorithms presented in Sections~\ref{sec:Hammerstein} and~\ref{sec:ParallelHammerstein}. \KoenSpellCheck[It is described in detail in \citep{SchoukensM2015a,SchoukensM2015c}.]{\citet{SchoukensM2015a,SchoukensM2015c} describe this in detail.}
		
		Firstly, the BLA decomposition approach is applied as described in Section~\ref{sec:ParallelHammerstein} for the parallel Hammerstein and parallel Wiener case to split the BLA dynamics over the different parallel branches.
		
		Next, the poles and zeros of each branch are allocated to the front or the back simultaneously using the brute-force pole-zero allocation scan described in Section~\ref{sec:Hammerstein}. The static nonlinearity of the model is estimated during this step.
		
		Although \KoenSpellCheck[it is shown in \citep{SchoukensM2015a,SchoukensM2015c}]{\citet{SchoukensM2015a,SchoukensM2015c} show} that this approach results in a consistent estimate of the system parameters, a nonlinear optimization of all the parameters simultaneously as described in Section~\ref{sec:ML} can be used to reduce the variance of the estimates.
		
		Two disadvantages of this approach are \M{i)} the high computational load due to the brute-force pole-zero allocation scan, and \M{ii)} the introduction of a MIMO static nonlinearity in the final model although the original system is \M{intrinsically} decoupled.\M[ of nature.]{} The appearance of the MIMO nonlinearity is due to equivalence transforms and the SVD approach to split the dynamics \citep{SchoukensM2015a,SchoukensM2015c}. The brute-force pole-zero allocation scan can be made faster by allocating the poles and zeros of each branch separately instead of doing the allocation for all branches simultaneously \citep{SchoukensM2015c}. The coupled, MIMO nonlinearity can be decoupled into a number of SISO (single input single output) static nonlinearities in a second step using tensor decomposition methods as \KoenSpellCheck[is shown in \citep{SchoukensM2012a,Tiels2013,Dreesen2015}]{\citet{SchoukensM2012a,Tiels2013,Dreesen2015} show}.
		
\section{Feedback Model Structures} \label{sec:Feedback}
	\Koen{The use of a feedback structure can be motivated by the presence of nonlinear feedback in the system under test, e.g. many mechanical setups show resonant behavior in which the resonance frequency depends on the variance of the excitation signal. As discussed later in this section, this is a typical nonlinear feedback phenomenon. Feedback mechanisms are also often encountered in biological systems. A nonlinear feedback can also be included in the model structure to increase its modeling capacity.}

	The BLA or $\epsilon$-approximation can be used in many cases for the identification of feedback block-oriented \Koen{nonlinear} model structures. Both frameworks allow the user to perform the selection of the number of poles and zeros in the system dynamics in the LTI framework, in the three cases considered in this section. Note that all feedback systems considered in this section are \Koen[considered to be]{} BIBO stable for the selected input class due to the \Koen[considered]{assumed} class of systems in Section~\ref{sec:BLA}.
    \M[No special attention is given here during identification to obtain stable models.]{During the identification procedure, no special effort is made to obtain stable models.} \Koen{Multiple shooting techniques can be used to pass through unstable regions in the parameter space \citep{Mulders2010}. Since all the considered feedback structures can be restructured so that they include the simple nonlinear feedback structure (see Figures~\ref{fig:Structures} and \ref{fig:IndistFB}), stability of the final model can be checked with e.g. the Popov criterion \citep{Haddad1993,Haddad1994,Khalil1996}. This criterion provides a sufficient (not necessary) condition for the overall model to be absolutely stable. A typical assumption is that the linear subblock is observable and controllable, and that the nonlinearity fulfills a sector condition.}
	
	It is commonly assumed for nonlinear feedback systems to have at least one sample delay in either the forward or the backwards path of the feedback loop to avoid the presence of nonlinear algebraic loops during simulation. It is shown in \citep{Relan2016} under band-limited assumptions that the approximation error made by introducing such a delay in the model can be kept arbitrary low by selecting a sufficiently high sampling frequency.
	
\subsection{Simple Feedback Structure} \label{sec:SimpleFB}
	The simple feedback structure consists of the feedback connection of an LTI block and a static nonlinear block (see Figure~\ref{fig:Structures}). One can think of two possible simple feedback structures: LTI block in the feedforward and the nonlinearity in feedback or vice-versa. It can be shown that both combinations are equivalent from an input-output point of view (see Section~\ref{sec:Structures} and \citep{Schoukens2008}). \KoenSpellCheck[A BLA-based identification algorithm for the simple feedback structure is proposed in \citep{Paduart2004,Paduart2008}.]{\citet{Paduart2004,Paduart2008} propose a BLA-based identification algorithm for the simple feedback structure.} \KoenSpellCheck[A slightly more advanced nonlinear {\Koen[feedack]{feedback}} structure is considered in \citep{Pelt2001}.]{\citet{Pelt2001} consider a slightly more advanced nonlinear feedback structure.}
	
	\subsubsection*{Best Linear Approximation}
		Contrary to most of the previous cases, the BLA of a simple feedback system cannot be simplified to an easy function of the system dynamics, even if input signals belonging to \Maarten[the Riemann equivalence class of asymptotically normally distributed excitation signals]{$\mathbb{S}_{U}$} are used. The input signal of the nonlinearity $y(t)$ is not Gaussian due to the presence of the nonlinear feedback loop. Hence, Bussgang's theorem cannot be applied. However, in many practical cases it is fair to assume that the nonlinearity can be linearized to a gain $\alpha$ and thus that the BLA will be approximately given by:
		\begin{align}
			G_{bla}(e^{j\omega T_s}) \approx \frac{G(e^{j\omega T_s})}{1+\alpha G(e^{j\omega T_s})}, \quad \alpha \in \mathbb{R}. \label{eq:BlaSimpleFeedback}
		\end{align}
		\Koen{The quality of this approximation depends on how non-Gaussian the input to the nonlinearity is. The approximation is relevant if the input to the nonlinearity is e.g. the output of a linear block with a sufficiently large memory length. It then follows from the central limit theorem that the input to the nonlinearity resembles a Gaussian \M{signal}. The approximation in \eqref{eq:BlaSimpleFeedback} is relevant if the memory length of $G(q)$ is sufficiently large.}
		
		The $\epsilon$-approximation, contrary to the BLA, results in a simple and exact analytical expression of the linear approximation as a function of the system dynamics \citep{Schoukens2015}. \Koen{The expression for $G_{\epsilon}(e^{j\omega T_s})$ is as in \eqref{eq:BlaSimpleFeedback} with $\alpha$ replaced by $\gamma$.}
	
		The gains $\alpha$ and $\gamma$ \M[depend not only]{do not only depend} on the system characteristics, but they can also depend on the variance, power spectrum, and the offset (DC value) of the input signal.
	
		Note that in the feedback case (also in the Wiener-Hammerstein feedback case in Section~\ref{sec:WienerHammersteinFB}), only the poles of the $\epsilon$-approximation shift with a changing gain $\gamma$:
		\begin{align} \label{eq:BlaCommonNum}
			G_{\epsilon}(q) = \frac{B(q)}{A(q)+\gamma B(q)}, 
		\end{align}
		where
		\begin{align}
			G(q) &= \frac{B(q)}{A(q)}.
		\end{align}
		\KoenSpellCheck[This is discussed in further detail in \citep{Schoukens2015}.]{\citet{Schoukens2015} discuss this in further detail.}	
		
	\subsubsection*{Identification using the BLA}
		The identification algorithm that is presented here follows the lines of the algorithm presented in \citep{Paduart2004,Paduart2008}.
		
		Firstly, it is important to note that an arbitrary gain $\beta$ can be shifted from the static nonlinear block to the linear block in the model. Indeed we have that:
		\begin{align}
			y(t) &= G(q)[u(t)-f(y(t))] \nonumber \\
					 &= G(q)[u(t)-f(y(t))+\beta y(t) - \beta y(t)] \nonumber \\
			(1-\beta G(q))[y(t)] &= G(q)[u(t)-\tilde{f}(y(t))] \\
			y(t) &= \tilde{G}(q)[u(t)-\tilde{f}(y(t))] \nonumber 
		\end{align}
		where $\tilde{f}(y(t)) = f(y(t))-\beta y(t)$ and \Koen[$\tilde{G}(q) = \frac{G(q)}{1-\beta G(q)}$]{$\tilde{G}(q) = \frac{G(q)}{1+\beta G(q)}$}.
		
		This indicates that the BLA results in a good initial estimate of the linear subsystem up to the equivalence transforms present in the model structure:
		\begin{align}
			\hat{G}(q) = G_{bla}(q) \approx \frac{G(q)}{1+\alpha G(q)}.
		\end{align} 
		
		\M{Next, }the static nonlinearity is estimated. The static nonlinearity is assumed to be represented as a linear combination of nonlinear basis functions. It is estimated by cutting the feedback loop:
		\begin{equation}\label{eq:NlIdentFB}\begin{aligned}
			\hat{G}(q)[u(t)]-y(t) &= \hat{G}(q)[\sum_{i=1}^{n_f} \gamma_i f_i(y(t))],  \\
													  &= \sum_{i=1}^{n_f} \gamma_i \hat{G}(q)[f_i(y(t))], 
		\end{aligned}\end{equation}
		where $n_f$ denotes the number of basis functions that is used to describe the static nonlinearity. \M[Note that this]{This} equation is linear in the parameters $\gamma_i$. Hence, the parameters $\gamma_i$ can easily be obtained by a linear least squares estimation. 
		
		\M[Note that the]{The} noisy output is used as a regressor during the initial identification of the static nonlinearity in eq.~\eqref{eq:NlIdentFB}. This results in a very simple initialization scheme which works well for high signal to noise ratios.
		
		Finally, a nonlinear optimization (see Section~\ref{sec:ML}) of all the parameters of the simple feedback structure together is applied such that a consistent estimate is obtained. \Koen{\M[Note that the]{The} noise is assumed \M{to be present} at the output of the system (Assumption~\ref{ass:Noise}), i.e. outside the feedback loop. The feedback mechanism acts on the noise-free signal.} The construction of the Jacobian of the simple feedback structure is more involved and time consuming (compared to the single and parallel branch case) due to its recursive dependency on the output at previous time instances. This holds as well for the Wiener-Hammerstein feedback case and the LFR case in Sections~\ref{sec:WienerHammersteinFB} and~\ref{sec:LFR}.

\subsection{Wiener-Hammerstein Feedback \Maarten[]{Structure}} \label{sec:WienerHammersteinFB}
	The Wiener-Hammerstein feedback structure is a generalized version of the simple feedback structure. It consists of a Wiener-Hammerstein system in the feedforward path combined with an LTI system in the feedback path (see Figure~\ref{fig:Structures}). Both BLA-based approaches \citep{Schoukens2008} and a convex relaxation\KoenSpellCheck{-}based approach \citep{Sou2008} have been developed for this system structure. It can be shown that this structure is equivalent to an LTI system in the forward path and a Wiener-Hammerstein system in the feedback path \citep{Schoukens2008}. 
	
	\subsubsection*{Best Linear Approximation}
		As for the simple feedback structure case, the BLA of a Wiener-Hammerstein feedback structure cannot be simplified to an easy function of the system dynamics if input signals belonging to \Maarten[the Riemann equivalence class of asymptotically normally distributed excitation signals]{$\mathbb{S}_{U}$} are used. However, in many practical cases it is fair to assume that the BLA will be approximately of the form:
		\begin{align}
			G_{bla}(e^{j\omega T_s}) &\approx \frac{\alpha G^{[1]}(e^{j\omega T_s})G^{[2]}(e^{j\omega T_s})}{1+\alpha G^{[1]}(e^{j\omega T_s})G^{[2]}(e^{j\omega T_s})G^{[3]}(e^{j\omega T_s})} \nonumber \\
															 &\alpha \in \mathbb{R}. \label{eq:BlaWienerHammersteinFB}
		\end{align}
        \Koen{The approximation in \eqref{eq:BlaWienerHammersteinFB} is relevant if the memory length of $G^{[2]}(q)G^{[3]}(q)G^{[1]}(q)$ is sufficiently large.}
		
		The $\epsilon$-approximation, contrary to the BLA, results in a simple and exact analytical expression of the linear approximation in function of the system dynamics \citep{Schoukens2015}\Koen[:]{. The expression for $G_{\epsilon}(e^{j\omega T_s})$ is as in \eqref{eq:BlaWienerHammersteinFB} with $\alpha$ replaced by $\gamma$.}
        \Koen[vergelijking hieronder in commentaar]{}
	
		The gains $\alpha$ and $\gamma$ depend not only on the system characteristics, but can also depend on the variance, power spectrum, and the offset (DC value) of the input signal.
		
	\subsubsection*{Identification using the BLA}
		\KoenSpellCheck[The BLA-based identification algorithm is outlined in \citep{Schoukens2008}.]{\citet{Schoukens2008} outline the BLA-based identification algorithm.} The \M{key observation is that} the inverse of the BLA is approximately given by:
		\begin{align}
			\frac{1}{G_{bla}(q)} \approx G^{[3]}(q) + \frac{\beta}{ G^{[1]}(q)G^{[2]}(q)}, \label{eq:BlaInvWienerHammersteinFB}
		\end{align}
		where $\beta = \frac{1}{\alpha}$.
		
		In a first step, the BLA of the system is identified at \Koen[two or more]{$R \ge 2$} different setpoints of the system:
		\begin{align}
			\hat{G}_{bla}^{[j]}(q) &\approx \frac{\alpha^{[j]} G^{[1]}(q)G^{[2]}(q)}{1+\alpha^{[j]} G^{[1]}(q)G^{[2]}(q)G^{[3]}(q)}.
		\end{align}
        \Koen{for $j = 1, \ldots, R$.}
		

		A nonparametric initial estimate $\hat{G}^{[3]}(e^{j\omega T_s})$ of the linear feedback system is obtained as the mean of the inverses of the BLA:
		\begin{align}
			\hat{G}^{[3]}(e^{j\omega_k T_s}) &= \frac{1}{R} \sum_{j=1}^R \frac{1}{\hat{G}_{bla}^{[j]}(e^{j\omega_k T_s})} \\
											 								 &\approx G^{[3]}(e^{j\omega_k T_s}) + \frac{\gamma}{G^{[1]}(e^{j\omega_k T_s})G^{[2]}(e^{j\omega_k T_s})}, \nonumber
		\end{align}
		where $\gamma = \frac{1}{R} \sum_{j=1}^R \beta^{[j]}$. 
		
		\Koen[A nonparametric initial estimate of $\hat{S}(e^{j\omega T_s}) \approx G^{[1]}(e^{j\omega T_s})G^{[2]}(e^{j\omega T_s})$ is obtained as follows:]{A nonparametric estimate of a scaled version of the inverse of $S(e^{j\omega T_s}) = G^{[1]}(e^{j\omega T_s})G^{[2]}(e^{j\omega T_s})$ is obtained at each setpoint:}
		\begin{align}
			E^{[j]}(e^{j\omega_k T_s}) = \frac{1}{\hat{G}_{bla}^{[j]}(e^{j\omega_k T_s})} - \hat{G}^{[3]}(e^{j\omega_k T_s}).
		\end{align}
    	\Koen{These estimates are collected for all frequencies and setpoints in a matrix as follows:}
		\begin{align}
			\boldsymbol{E} = \left[ \begin{array}{c c c}
													E^{[1]}(e^{j\omega_1 T_s}) & \ldots & E^{[R]}(e^{j\omega_1 T_s}) \\
													\vdots										 & \ddots & \vdots										 \\
													E^{[1]}(e^{j\omega_F T_s}) & \ldots & E^{[R]}(e^{j\omega_F T_s}) \\
											 \end{array} \right].
		\end{align}
		The SVD of this matrix results in:
		\begin{align}
			\boldsymbol{E} = \boldsymbol{U} \boldsymbol{\Sigma} \boldsymbol{V}^T.
		\end{align}
		The estimate $\hat{S}^{-1}(e^{j\omega_k T_s})$ is \Koen[a scaled version]{}equal to the first left singular vector present in the matrix $\boldsymbol{U}$. \Koen[This scaling factor can be ignored, it will be accounted for in the estimate of the static nonlinear block.]{Taking its inverse and ignoring the scaling factor (it will be accounted for in the estimate of the static nonlinear block), an estimate of the product $G^{[1]}(e^{j\omega T_s})G^{[2]}(e^{j\omega T_s})$ is obtained.}
		
		Splitting the dynamics $\hat{S}(e^{j\omega_k T_s})$ to the front and the back linear block of the Wiener-Hammerstein system can be done using a combination of the methods presented in Sections~\ref{sec:Hammerstein} and~\ref{sec:SimpleFB}. This results in the estimates $\hat{G}^{[1]}(e^{j\omega_k T_s})$, $\hat{G}^{[2]}(e^{j\omega_k T_s})$ and $\hat{f}(.)$.
		
		Finally, a nonlinear optimization (see Section~\ref{sec:ML}) of all the parameters of the closed loop Wiener-Hammerstein feedback structure together is applied such that a consistent estimate is obtained. \Koen{Note from Assumption~\ref{ass:Noise} that the feedback mechanism acts on the noise-free signal.}
		
\subsection{Linear Fractional Representation} \label{sec:LFR}
	\KoenSpellCheck[The linear fractional representation (LFR) model structure is studied in \citep{Vandersteen1999,Hsu2008,Novara2011,Mulders2013,Vanbeylen2013}.]{\citet{Vandersteen1999,Hsu2008,Novara2011,Mulders2013,Vanbeylen2013} study the linear fractional representation (LFR) model structure.} It is a combination of a Wiener-Hammerstein feedback system with a linear feedforward branch in parallel (see Figure~\ref{fig:Structures}). It is the most general block-oriented \Koen{nonlinear} system that contains only one static nonlinear block. All other block-oriented \Koen{nonlinear} systems with only one static nonlinearity can be reformulated \M[to]{as} an LFR structured system.
	
	The approach that is presented in \citep{Vandersteen1999} relies on a large set of two-tone measurements and Volterra theory to estimate the LFR model. The method presented by \citep{Mulders2013} starts from a polynomial nonlinear state space model and transforms it into a LFR model. The block-oriented, BLA-based method that is presented in \citep{Vanbeylen2013} is studied in some more detail below.
	
	\KoenSpellCheck[An extension to LFR models with a MIMO static nonlinear block is presented in \citep{Hsu2008,Novara2011}.]{\citet{Hsu2008,Novara2011} present an extension to an LFR model with a MIMO static nonlinear block.} This LFR structure with MIMO blocks includes both filtered output and process noise. While the methods presented in \citep{Vandersteen1999,Mulders2013,Vanbeylen2013} start from input-output data only, the methods presented in \citep{Hsu2008,Novara2011} assume that the inputs to the nonlinear block are measured (i.e. they can be obtained from the measured inputs and outputs and the known linear block). Furthermore, the static nonlinear block is identified under the assumptions that the MIMO LTI block is known and that the nonlinear block has a known block-diagonal structure with MISO blocks. 
	
	\subsubsection*{Best Linear Approximation}
	The BLA of a general LFR system cannot be simplified to an easy function of the system dynamics, even if input signals belonging to \Maarten[the Riemann equivalence class of asymptotically normally distributed excitation signals]{$\mathbb{S}_{U}$} are used. The input signal of the nonlinearity $y(t)$ is not Gaussian due to the presence of the nonlinear feedback loop. Hence, Bussgang's theorem cannot be applied. However, in many practical cases it is fair to assume that the nonlinearity can be linearized to a gain $\alpha$ and thus that the BLA will be approximately of the form:
		\begin{align}
			G_{bla}(e^{j\omega T_s}) &\approx  G^{[1]}(e^{j\omega T_s}) + \frac{\alpha G^{[2]}(e^{j\omega T_s})G^{[3]}(e^{j\omega T_s})}{1+\alpha G^{[4]}(e^{j\omega T_s})}, \nonumber \\
			& \quad \alpha \in \mathbb{R}. \label{eq:BlaLFR}
		\end{align}
        \Koen{The approximation in \eqref{eq:BlaLFR} is relevant if the memory length of $G^{[4]}(q)$ is sufficiently large.}
		
		The $\epsilon$-approximation, contrary to the BLA, results in a simple and exact analytical expression of the linear approximation in function of the system dynamics \citep{Schoukens2015}\Koen[:]{.}
        \Koen[vergelijking hieronder in commentaar]{The expression for $G_{\epsilon}(e^{j\omega T_s})$ is as in \eqref{eq:BlaLFR} with $\alpha$ replaced by $\gamma$.}
	
		The gains $\alpha$ and $\gamma$ \M[depend not only]{do not only depend} on the system characteristics, but they can also depend on the variance, power spectrum, and the offset (DC value) of the input signal.
	
		Note that in the LFR case both the poles and zeros of the $\epsilon$-approximation shift with a changing gain $\gamma$:
		\begin{align}
			G_{\epsilon}(q) = G^{[1]}(q) + \frac{\gamma G^{[2]}(q)G^{[3]}(q)}{1+\gamma G^{[4]}(q)}, 
		\end{align}
		where
		\begin{align}
			G^{[i]}(q) &= \frac{B^{[i]}(q)}{A^{[i]}(q)}.
		\end{align}
		\Koen[This results in the following expression,]{Introducing a common denominator results in an expression} containing a dependency on $\gamma$ in both numerator and denominator\Koen[:]{.}
        \Koen[vergelijking hieronder in commentaar]{}
		\KoenSpellCheck[This is discussed in further detail in \citep{Schoukens2015}.]{\citet{Schoukens2015} discuss this in further detail.}	
	
	\subsubsection*{Identification using the BLA}
		The identification approach for the LFR model that is presented in \citep{Vanbeylen2013} shares some of its aspects with the approaches that are explained above: it relies on at least two measured BLAs, and it needs to perform a scan over several possible solutions to allocate the dynamics over the different LTI-blocks in the model. The main aspects of the identification algorithm are described below.
		
		First, the BLA is estimated at two different setpoints of the system. This BLA is converted to linear state space models.
		
		Next, an equivalence transformation for each state space model is estimated by solving a nonsymmetric algebraic Riccati equation. The details of this equation can be found in \citep{Vanbeylen2013}. This results in different possible solutions. 
		
		Each solution is evaluated by constructing the corresponding LFR model. The model with the best fit is selected. This step is very similar to the pole-zero allocation scan during the Wiener-Hammerstein identification in Section~\ref{sec:Hammerstein}.
		
		A nonlinear optimization (see Section~\ref{sec:ML}) of all the parameters of the closed loop simple feedback structure together is applied such that a consistent estimate is obtained. \Koen{Note from Assumption~\ref{ass:Noise} that the feedback mechanism acts on the noise-free signal.}

\section{Pros and cons of the model structures} \label{sec:Overview}
The model structures and identification methods that are explained in the previous sections all have their pros and cons.
Some model structures have the ability to accurately model the input-output behavior of a large class of systems, but use a large number of parameters in return.
Some identification methods can easily generate starting values from a single experiment, others require computationally expensive steps on multiple data sets.
This section provides an overview of the model structures in terms of their descriptive power, number of parameters, and general difficulty \KoenSpellCheck[to estimate]{in estimating} the model.
These criteria are explained below. \Koen{Next, the different model structures are compared based on these criteria, and some general limitations of BLA-based block-oriented nonlinear modeling are discussed.}

\Koen{\subsection{Criteria for comparison}}
	In the context of this paper, a universal approximator is a model structure that can approximate arbitrarily well the input-output behavior of systems belonging to the \Koen{PISPO system class (see Section~\ref{sec:SystemClass})} in mean square sense.
Parallel Wiener and parallel Wiener-Hammerstein models are universal approximators without any constraint on the excitation class \citep{Schetzen2006}.
Other approximation properties of these model structures have been shown in the past \citep{Palm1979,Boyd1985,Korenberg1991}.

	A large number of independent parameters leads to a large variance on the estimated parameters \citep{Soderstrom1989} when a model is identified from noisy measurements.
Therefore, a model is preferably parameter-parsimonious, as decreasing the variability on the estimated parameters requires larger data sets, thereby increasing the computational time of the optimization step in Section~\ref{sec:ML}.

	The optimization step in Section~\ref{sec:ML} converges to a local optimum of the (non-convex) cost function. Good initial estimates are thus required. A consistent initial estimate is preferable, as in that case a good initial estimate is guaranteed as more estimation data are used.
\Koen[The optimization step itself is consistent for each of the considered model structures.]{As described in Section~\ref{sec:ML}, the nonlinear estimator is consistent under a limited set of assumptions for each of the considered model structures.}

	Some identification approaches require multiple data sets, e.g. BLA estimates at different setpoints, which is not always possible due to time constraints.

\Koen{\subsection{Comparison of model structures}}
	Table~\ref{tbl:overview} provides an overview of \M[these properties]{the criteria for comparison} as well as some general difficulties of the identification algorithms that are presented in the previous sections.

	Only the parallel Wiener, Wiener-Schetzen, and parallel Wiener-Hammerstein model structures are universal approximators. Nevertheless, these feedforward structures cannot capture chaotic behavior, sub-harmonics, or hysteresis, since this requires feedback. The Linear Fractional Representation (LFR) is the most general block-oriented structure that contains only one static nonlinear block. The Wiener, Hammerstein, Wiener-Hammerstein, simple feedback, and Wiener-Hammerstein feedback structures are thus all special cases of the LFR.

	The single branch models of Section~\ref{sec:Single} all have a low number of parameters, since these structures consist of a limited series connection of SISO blocks.
The parallel Hammerstein model structure consists of a number of parallel connections of Hammerstein models and thus still has a fairly low number of parameters.
The parallel Wiener and Wiener-Schetzen model structures both have a MISO static nonlinear block, which typically results in a high number of parameters. For the parallel Wiener structure, however, the proposed identification method results in a low number of parallel branches, i.e. a low number of inputs to the nonlinear block, and thus a medium number of parameters.
The parallel Wiener-Hammerstein structure has a MIMO static nonlinear block, resulting in a high number of parameters.
The feedback structures of Section~\ref{sec:Feedback} consist of a limited (at most five) SISO blocks, and thus have a low number of parameters. The LFR structure, however, is typically \KoenSpellCheck[parameterized]{parametrized} in state space form \citep{Vanbeylen2013}, which introduces at least $n^2$ redundant parameters, where $n$ is the model order. This is due to the existence of an $n \times n$ transformation matrix that leaves the input-output behavior unchanged. This transformation matrix should be nonsingular, but \M[otherwise has]{except for this constraint the} elements \M[that]{}can be freely chosen. The LFR structure can thus have a medium number of parameters, although the number of independent parameters is still low.

	Most of the proposed identification methods result in consistent initial estimates. This is not the case for the Hammerstein-Wiener identification method and the methods for the feedback structures of Section~\ref{sec:Feedback}, since these methods use the noise-corrupted output either in some of the regressors (Hammerstein-Wiener and simple feedback) or to reconstruct internal signals that constitute some of the regressors (Wiener-Hammerstein feedback and LFR). Moreover, an approximation to the BLA is made in the Hammerstein-Wiener structure and all feedback structures. This approximation can be avoided by replacing the BLA-framework by the theoretical framework of the $\epsilon$-approximation \citep{Schoukens2015}.

	The presented methods on model structures that include a Wiener-Hammerstein branch all involve a brute-force scan to separate the dynamics of the Wiener-Hammerstein branch over its two LTI blocks. A similar scan over the solutions of a nonsymmetric algebraic Riccati equation is performed in the LFR identification method. These scans are mildly computationally expensive if only one Wiener-Hammerstein branch is involved (Wiener-Hammerstein, Wiener-Hammerstein feedback, and LFR structures). The scan is computationally expensive for the parallel Wiener-Hammerstein structure, since there, a scan needs to be performed for each branch and each time a MIMO nonlinearity needs to be estimated.
    
    \Koen{The BIBO stability of the feedforward models is easy to check, since if all subblocks are BIBO stable, the overall model is stable, and the stability of LTI and SNL subblocks can be easily checked. For the feedback models, the stability checks are more involved (see Section~\ref{sec:Feedback}).}

\begin{table*}
\caption{Overview of the considered model structures and their identification using the BLA, as explained in Sections~\ref{sec:Single}--\ref{sec:Feedback}.
	Hammerstein is abbreviated as H, Wiener as W, and Linear Fractional Representation as LFR.
\label{tbl:overview}}
    \begin{tabular}{p{2.5cm} p{1cm} p{2cm} p{1.8cm} p{1.8cm} p{6.1cm}}
		\toprule
		Model structure	& Section 	& Universal approximator		& Number of parameters	& Consistent initial values	& General difficulties	\\
		\midrule
		Hammerstein (H)	& \textcolor{black}{\ref{sec:Hammerstein}}			& no			& low		& yes 		& / 						\\
		Wiener (W) 		& \textcolor{black}{\ref{sec:Hammerstein}}			& no			& low		& yes 		& / 						\\
        W-H 			& \textcolor{black}{\ref{sec:Hammerstein}}			& no			& low		& yes 		& mildly computationally expensive brute-force scan	\\
		H-W 			& \textcolor{black}{\ref{sec:HammersteinWiener}}	& no			& low		& no$^1$ 	& assumes an invertible output nonlinearity, iterative procedure					\\
		Parallel H 		& \textcolor{black}{\ref{sec:ParallelHammerstein}}	& no			& low/medium & yes 		& BLA needed at different setpoints 			\\
		Parallel W 		& \textcolor{black}{\ref{sec:ParallelHammerstein}}	& yes/no$^2$	& medium	& yes 		& BLA needed at different setpoints 			\\
		Wiener-Schetzen	& \textcolor{black}{\ref{sec:WienerSchetzen}}		& yes			& high$^3$	& yes 		& \textcolor{black}{number of parameters strongly depends on the quality of the selected basis functions} 						\\
		Parallel W-H 	& \textcolor{black}{\ref{sec:ParallelWienerHammerstein}} & yes/no$^2$	& high$^4$	& yes 	& BLA needed at different setpoints, computationally expensive brute-force scan 	\\
		Simple feedback & \textcolor{black}{\ref{sec:SimpleFB}}			& no			& low		& no 		& \textcolor{black}{stability nonlinear feedback hard to check} 						\\
		W-H feedback 	& \textcolor{black}{\ref{sec:WienerHammersteinFB}}	& no			& low		& no 		& BLA needed at different setpoints, mildly computationally expensive brute-force scan\Koen{, stability nonlinear feedback hard to check}					\\
		LFR 			& \textcolor{black}{\ref{sec:LFR}}					& no			& low/medium & no 		& BLA needed at different setpoints,			mildly computationally expensive scan over solutions of a NARE$^5$\Koen{, stability nonlinear feedback hard to check}			\\
		\bottomrule
	\end{tabular}
	\Maarten[]{$^1$: only consistent in absence of noise and model errors. 
	$^2$: although the model structure is a universal approximator, the proposed identification methods cannot deal with the general class of Volterra systems; some limitations apply. 
	$^3$: multiple input single output polynomial. 
	$^4$: multiple input multiple output static nonlinearity. 
	$^5$: nonsymmetric algebraic Riccati equation.}
\end{table*}

\Koen{
\subsection{Limitations BLA-based block-oriented nonlinear modeling}

From the comparison of the models in Table~\ref{tbl:overview}, it can be observed that the initialization can be cumbersome for the more complex model structures, while the flexibility of the \KoenSpellCheck[more simple]{simpler} models can be limited.

The BLA-based block-oriented approaches presented \M[here]{in this paper} are multistage approaches. In a first stage, the BLA is determined, and in one or more later stages, the full nonlinear model is estimated.
\Maarten[]{Very noisy data, or strongly nonlinear systems (e.g. a hard saturation nonlinearity) can result in a poor BLA estimate. \Koen[This BLA forms the cornerstone of the algorithms that are presented in this paper.]{} It is more likely that the nonlinear optimization will end up in a local minimum of the cost function in the noisy or very nonlinear case. However, we also know that more data will eventually result in better estimates since most of the initialization algorithms that are presented in this paper (and also the BLA estimator itself \citep{Pintelon2012}) are consistent estimators. Therefore, we advice the user to collect more data, if possible, in \Koen[the]{} case the measurements are (very) noisy or the system is strongly nonlinear. Using periodic excitations allows one to easily obtain an estimate of the level of the noise in the measurements, averaging over the periods allows one to decrease the influence of the noise on the measurements.}

A particular case is the presence of an even nonlinearity. As pointed out in Section~\ref{sec:BLA}, the BLA is then zero. The quality of the estimate of the linear block(s), on which the presented methods rely, is typically poor in this case.
An alternative for the BLA-based approach that overcomes this issue with an even nonlinearity is the weighted principal component analysis (PCA) method as reported \KoenSpellCheck[in]{by} \citet{Zhang2015} for the identification of Wiener systems.

}

\section{Guidelines for the User} \label{sec:Guidelines}

	Two important steps in the system identification process are the input design and model structure selection step. This section provides some user guidelines for \Koen[the input design and the model structure selection problem]{these two steps}.

	\subsection{Model Structure Selection}
	
		The use of the BLA framework is not limited to obtaining a block-oriented \Koen{nonlinear} model of a nonlinear system. The BLA framework also offers a powerful analysis tool that allows the user, when combined with a well thought out input design strategy, to acquire a lot of insight into the nonlinear system at very low cost. The questions \textit{\KoenSpellCheck["]{``}How nonlinear is the system?"}, \textit{\KoenSpellCheck["]{``}Do I need a nonlinear model?"}, and \textit{\KoenSpellCheck["]{``}What kind of nonlinear effects do I observe in the system?"} can be (partly) answered within the BLA framework \citep{Lauwers2008,Pintelon2012,Schoukens2012,Schoukens2015}, an overview can be found in \citep{Schoukens2016}.
		
		The BLA framework offers the user the possibility \KoenSpellCheck[to obtain]{of obtaining} both a nonparametric frequency domain estimate of the noise disturbance present at the output and \M{of} the system combined, and a nonparametric frequency domain estimate of the nonlinear disturbances generated by the system \citep{Pintelon2012,Schoukens2012}. These can be used to quantify the level of the nonlinearities in the system and \M{to determine} whether or not the nonlinearities or the noise are the dominant disturbance source. Hence, it also allows the user to make a choice on whether or not a nonlinear model should be estimated, based on a quantitative measure. \M[It could be that a linear model is accurate enough.]{In some cases, depending on the final goal, a linear model could be accurate enough.}
		
		Exciting the nonlinear system under test at different setpoints of the system (e.g. different input variances or mean values of the input signal) can be very useful to detect what kind of nonlinear effects are present in the system. It is presented in \citep{Schoukens2015} that a pole shift in the BLAs over the different setpoints indicates the presence of nonlinear feedback loops, while a zero shift indicates the presence of parallel nonlinear signal paths in the system. \Koen[It are exactly these]{These} pole and zero shifts \Koen[that]{}are exploited by some of the presented identification algorithms to generate initial estimates of the linear dynamic blocks of the selected block-oriented structure. The BLA-based model structure selection approach can be combined with other model structure selection approaches, e.g. \citep{Haber1990,Pearson2003}, to guide the user to a suitable model structure.

		It is often not possible to perform an elaborate measurement campaign on the nonlinear system that needs to be modeled. A thorough system analysis and model structure selection step are not always possible under these circumstances. Fortunately, most of the presented modeling approaches are relatively simple and can be applied at a fairly low computational cost. This allows the user to try out multiple model structures and compare the quality of the resulting model on a validation data set.
	
	\subsection{Input Design}

		A model is in many practical cases only an approximation of the system under study. Therefore, it is of great importance to use input signals that represent a realistic set of excitations to the system. The system should be excited within a realistic range of operation, both in amplitude and frequency. Also the amplitude distribution (e.g. the crest factor of the input) can play an important role. These considerations often exclude relatively simple excitation signals such as sine inputs and impulses.

		The BLA framework prefers, but is not limited to, different realizations of periodic, random signals to perform a nonparametric analysis of the nonlinear system under test. An increase in the number of signal realizations decreases the variance on the BLA estimate due to both the nonlinear and noise distortion, while an increase of the number of periods only decreases the variance due to the noise distortion. A typical choice would be to use 7 different realizations and 2 (steady state) periods of a periodic, random signal \citep{Pintelon2012,Schoukens2012}. Some of the presented identification algorithms require the BLA to be measured at different setpoints of the system. \KoenSpellCheck[An experimental algorithm is introduced in \citep{Esfahani2016} to find a set of input signal rms-values and DC offsets such that the distortion on the BLA is minimized.]{\citet{Esfahani2016} introduce an experimental algorithm to find a set of input signal rms-values and DC offsets such that the distortion on the BLA is minimized.}

		Multisine signals (and more specifically random phase multisines) are excellent candidate input signals since they are periodic and offer a full control of the amplitude spectrum of the input signal \citep{Pintelon2012}. Small multisine signals can be \M[superposed]{superimposed} on a more realistic input signal for the system under test to obtain a random input signal \M[containing]{with} a rich frequency content that can be used for a nonparametric analysis of the nonlinear system under test. An example of this can be found in \citep{Widanage2011,Marconato2014}.

		The identification algorithms that are presented in the previous sections make use of the Riemann equivalence class of asymptotically normally distributed excitation signals $\mathbb{S}_{U}$ (including random phase multisine signals). This could be seen as a limitation of the identification algorithms. However, the presented algorithms are quite robust to deviations of this signal class due to the final nonlinear optimization step described in Section~\ref{sec:ML}.

\section{Benchmark and Practical Results} \label{sec:Benchmark}
	This section provides an overview of the benchmark setups and real-life systems on which the BLA-based block-oriented identification methods have been applied \Koen{as well as the lessons learned from them}. These results illustrate the good performance and the practical usefulness of the BLA-based block-oriented nonlinear modeling approaches.
	
	The Wiener-Hammerstein identification techniques in Section~\ref{sec:Hammerstein} obtained excellent results on the Wiener-Hammerstein benchmark \citep{Schoukens2009a,Hjalmarsson2012,Sjoberg2012a}. The simple feedback structure and the Wiener-Hammerstein feedback structure have been successfully applied on an electronic simulator of a hardening spring available as the Silverbox benchmark \citep{Paduart2008,Wigren2013}. The coupled electric drives benchmark in \citep{Wigren2013} \Koen[are]{is} an example of a system that is difficult to identify using the BLA\Koen{-}based approaches due to the presence of the purely even nonlinearity (see Section~\ref{sec:BLA}).
	
	Other systems, not available as a benchmark, that have been successfully modeled by the block-oriented BLA-based approaches are: the insulin-glucose regulatory system (Wiener, Wiener-Schetzen and LFR models) \citep{Marconato2014,Vanbeylen2014a}, a	RF crystal detector (Wiener-Hammerstein feedback model) \citep{Schoukens2008}, a logarithmic amplifier (Wiener and parallel Wiener models) \citep{SchoukensM2012b}, a valve audio amplifier (Hammerstein and parallel Hammerstein models) \citep{SchoukensM2011} and a Doherty power amplifier (parallel Wiener-Hammerstein model) \citep{SchoukensM2015c}. \KoenSpellCheck[\Koen{Applications in biomedical engineering and physiology using related correlation-based approaches are reported in \cite{Westwick2003}.}]{\citet{Westwick2003} report applications in biomedical engineering and physiology using related correlation-based approaches.}
    
    \Koen{From these and other practical results, it was found that including structural prior knowledge leads to a better trade-off between modeling performance and parameter parsimony \citep{Marconato2013,Vanbeylen2014a}. Nevertheless, the selection of the model structure is still a user choice in which the intended application of the model plays a role. If the model will be used for control purposes for example, it makes sense to restrict the model class to invertible models \citep{Pajunen1992,Gaasbeek2016}. If the model needs to be re-estimated often on new datasets, it makes sense to use a model structure that can be easily identified and that still has an acceptable error performance, even though it is known that the model is too simple \citep{Rebillat2016}.}
	
\section{Future Research Directions} \label{sec:futureWork}

	This overview paper surveyed a quite broad range of block-oriented identification approaches, with the emphasis on methods starting from linear approximations of the nonlinear system, for a wide range of block-oriented structures. It is also clear from the surveyed work that there are quite some open problems in the identification of block-oriented \Koen{nonlinear} systems. Some of these open problems are listed in this section. \Koen{Quite some of these open problems are already addressed to some extent, but the full picture is not \M{yet} complete.}

	A first open research question is how to include dynamic nonlinearities in the block-oriented \Koen{nonlinear} modeling toolbox. Quite some research has already been performed on this topic, for backlash, backlash-inverse and hysteresis nonlinearities \citep{Rochdi2010,Giri2011,Wang2012,Giri2013,Giri2014,Brouri2014b}. First steps to include dynamic nonlinearities in a BLA-based identification framework are made in \citep{Yong2015} for Hammerstein and Wiener model structures. Most of the methods that are developed until now consider Hammerstein and Wiener structures only, and are often limited by a specific input signal class. Further research on this topic is encouraged through the hysteretic nonlinear system identification benchmark presented in \citep{Noel2016}.
	
	The extension of the available identification methods to more complex structures is another open research question. This includes for instance the development of flexible identification methods for complex, user-generated and problem-dependent block-oriented structures and the extension of the existing block-oriented structures towards a MIMO setting without losing the structured nature of a block-oriented \Koen{nonlinear} model. Some efforts have already been made to identify MIMO block-oriented structures, both for the \KoenSpellCheck[more simple]{simpler} Hammerstein and Wiener block-oriented structures \citep{Verhaegen1996,Goethals2005b,Westwick1996,Janczak2007}, as for the more involved Wiener-Hammerstein, parallel Hammerstein and Wiener-Schetzen model structures \citep{Boutayeb1995,Katayama2016,SchoukensM2011a,Tiels2012}. These results could serve as a starting point for the development of other MIMO block-oriented identification algorithms.
	
	Most block-oriented identification methods assume the presence of an output-error noise framework. This is often a simplified representation of reality, as discussed in Section~\ref{sec:Structures}. Developing identification methods that are able to deal with more complex noise frameworks is an ongoing research effort (see for instance \citep{Hagenblad2008,Wills2013,Lindsten2013,Wahlberg2014,Mu2014} and also the results on the benchmark presented in \citep{SchoukensM2016a}). 
	
	A block-oriented \Koen{nonlinear} model is often an approximation of the true system behavior. How to deal with model errors in nonlinear system identification, and characterizing these model errors is an open research problem. The results of the BLA framework \citep{Pintelon2012} could help as a guideline in how to develop a nonlinear system identification framework dealing with model errors.
	
	In many practical cases, the user has some prior knowledge about the system behavior. This can be information about the system structure and the type of the static nonlinearity (saturation, discontinuity, smoothness), but also about the location of some poles and zeros and the stable nature of the system. Ideally, such prior knowledge should be taken into account during the identification of the system. The recently developed regularization\KoenSpellCheck{-}based identification methods allow \M{one} to include such prior knowledge (as is illustrated in \citep{Risuleo2015} for the Hammerstein model structure and \citep{Birpoutsoukis2015} for Volterra kernel estimation). \M[and the]{The} work presented in \citep{Tiels2014b} \M[present]{proposes} an approach to include prior pole-zero information during the identification for Wiener-Schetzen model structures.
	
	The input of a nonlinear system is not always controlled by the user (e.g. when \Koen[we are]{} dealing with rainfall input), or it can be very costly to collect a high number of data points. The identification of poorly excited systems, or the identification of systems with short data records is an open problem (see also the benchmark presented in \citep{SchoukensM2016b}). Regularization\KoenSpellCheck{-}based identification methods are a possible solution to this problem (e.g. \citep{Risuleo2015}). However, more research has to be performed to adapt these approaches for nonlinear system identification purposes.
	
	Finally, an often neglected or underestimated problem in the development of a nonlinear system identification algorithm is its user friendliness. Most users are not \Koen[an expert]{experts} in (nonlinear) system identification. Therefore, it is of great importance to develop identification algorithms for complex block-oriented structures that are robust with respect to the user choices, require as little user interaction as possible, and have little restrictions on the data used for identification. Such identification algorithms start to be available for the \KoenSpellCheck[more simple]{simpler} block-oriented structures (e.g. Hammerstein and Wiener), but this is typically not yet the case for the more complex block-oriented structures (e.g. parallel Wiener-Hammerstein, Wiener-Hammerstein feedback, LFR). Developing algorithms that have less restrictions on the estimation data, and that are more user friendly in general could increase the applicability of such model structures significantly.

\section{Conclusion} \label{sec:conclusion}
	This paper \Koen[gives an overview of the]{surveys} different block-oriented \Koen{nonlinear} systems that can be identified based on the best linear approximation, and it gives an overview of the identification algorithms that have been developed in the past. It turns out that a wide range of systems can be modeled using the best linear approximation approaches: single branch systems, parallel branch systems and nonlinear feedback systems. Furthermore, some future and ongoing research directions are discussed.
	
	The best linear approximation framework allows the user to extract important information about the system (e.g. model order of the dynamics, system structure information), it guides the user in taking good modeling decisions, and it proves to be a good starting point for nonlinear system identification algorithms for a wide range of block-oriented \Koen{nonlinear} model structures. The best linear approximation provides in some cases a direct estimate of the linear dynamic block of the system. In other cases the best linear approximation can be combined with a pole-zero allocation scan, a singular value decomposition, orthogonal basis functions or a Riccati-equation\KoenSpellCheck{-}based approach to identify the nonlinear system under test.

\begin{ack} 
	 This work was supported in part by the research council of the VUB (OZR bridging mandate), by the Belgian Government through the Inter university Poles of Attraction IAP VII/19 DYSCO program, and the ERC advanced grant SNLSID, under contract 320378.
\end{ack}	

\bibliographystyle{model5-names} 													 
\bibliography{ReferencesLibraryV2}    
		
\end{document}